\documentclass[journal]{IEEEtran}
\usepackage[utf8]{inputenc}
\usepackage[pdftex]{graphicx}
\usepackage{epstopdf}
\usepackage[utf8]{inputenc} 
\usepackage[T1]{fontenc}    
\usepackage{hyperref}       
\usepackage{url}            
\usepackage{booktabs}       
\usepackage{amsfonts}       
\usepackage{nicefrac}       
\usepackage{microtype}      
\usepackage{amsmath,graphicx}
\usepackage{amsbsy}
\usepackage{epsfig}
\usepackage{float}
\usepackage{graphicx}
\usepackage{color}
\usepackage{url}
\usepackage{cite}
\usepackage{amssymb}
\usepackage{pgf, tikz, pgfplots, epstopdf}
\usepackage{graphicx}
\usepackage{lipsum}
\usepackage{float}
\usepackage{cuted}
\usepackage{balance}
\usepackage{booktabs}
\usepackage{verbatim}
\usepackage{hyperref}
\usepackage{pgfplots}
\usepackage{calc}  
\usepackage{mathtools}
\usepackage{enumitem}  
\usepackage{multirow}
\newfloat{algorithm}{tbp}{loa}
\providecommand{\algorithmname}{Algorithm}
\usepackage{algorithm}
\usepackage[noend]{algpseudocode}
\usepackage{subcaption}
\input{mysymbol.sty}
\usepackage{stackengine}
\newcommand\xrowht[2][0]{\addstackgap[.5\dimexpr#2\relax]{\vphantom{#1}}}
\usepackage{enumitem}

\makeatletter
\newcommand*\dotp{\mathpalette\dotp@{.5}}
\newcommand*\dotp@[2]{\mathbin{\vcenter{\hbox{\scalebox{#2}{$\m@th#1\bullet$}}}}}
\newcommand\svdeq{\stackrel{\mathclap{\scriptsize\mbox{ SVD }}}{=}}

\makeatother
\def\x{{\mathbf x}}
\def\f{{\mathbf f}}
\def\g{{\mathbf g}}

\def\F{{\mathbf F}}
\def\U{{\mathbf U}}
\def\G{{\mathbf G}}
\def\D{{\mathbf D}}
\def\S{{\mathbf S}}
\def\B{{\mathbf B}}
\def\R{{\mathbf R}}
\def\h{{\mathbf h}}

\def\tmx{{\mathcal{T}_x\mathcal{M}}}
\def\tmxrp{{\mathcal{T}_x\mathbb{R}^p}}

\def\tm{{\mathcal{T}\mathcal{M}}}
\def\M{{\mathcal{M}}}
\def\i{\iota}
\def\rp{{\mathbb{R}^p}}
\makeatletter
\newcommand{\inc}{%
  \mathrel{\mathpalette\inc@\relax}%
}
\newcommand{\inc@}[2]{%
  \sbox\z@{$#1\lhd$}%
  \sbox\tw@{$#1\leqslant$}%
  \dimen@=\ht\tw@
  \advance\dimen@-\ht\z@
  \ifx#1\displaystyle
    \advance\dimen@ .2pt
  \else
    \ifx#1\textstyle
      \advance\dimen@ .2pt
    \fi
  \fi
  \ooalign{\raisebox{\dimen@}{$\m@th#1\lhd$}\cr$\m@th#1\leqslant$\cr}%
}
\makeatother
\makeatletter
\newcommand{\coinc}{%
  \mathrel{\mathpalette\coinc@\relax}%
}
\newcommand{\coinc@}[2]{%
  \sbox\z@{$#1\rhd$}%
  \sbox\tw@{$#1\geqslant$}%
  \dimen@=\ht\tw@
  \advance\dimen@-\ht\z@
  \ifx#1\displaystyle
    \advance\dimen@ .2pt
  \else
    \ifx#1\textstyle
      \advance\dimen@ .2pt
    \fi
  \fi
  \ooalign{\raisebox{\dimen@}{$\m@th#1\rhd$}\cr$\m@th#1\geqslant$\cr}%
}
\makeatother

\def\ltm{{\Gamma(\tm)}}
\def\ltmn{{\Gamma(\tm_n)}}

\def \bphi{{\boldsymbol{\phi}}}
\def \Phii{{\boldsymbol{\phi}_i}}
\def \th{{\widetilde{h}}}
\def \bPsi{{\boldsymbol{\Psi}}}
\def\Oi{{\mathbf{O}_i}}
\def\Oj{{\mathbf{O}_j}}
\def\Oij{{\mathbf{O}_{i,j}}}
\def\hd{{\hat{d}}}
\def \samp{\boldsymbol{\Omega}}

\def \eigGammai{\widetilde{\boldsymbol{\phi}}^n_{i}}
\def \eivGammai{\widetilde{\lambda}^n_{i}}
\floatstyle{ruled}
\newfloat{algorithm}{tbp}{loa}
\providecommand{\algorithmname}{Algorithm}
\floatname{algorithm}{\protect\algorithmname}

\usepackage{amsthm}
\newtheoremstyle{claudio}
  {-0.1em}
  {-0.1em}
  {}
  {}
  {\itshape\bfseries}
  {.}
  { }
  {\thmname{#1}\thmnumber{ #2 }\thmnote{(#3)}}

\theoremstyle{claudio}
\newtheorem{theorem}{Theorem}
\newtheorem{definition}{Definition}
\newtheorem{proposition}{Proposition}
\newtheorem{remark}{Remark}
\title{Tangent Bundle Convolutional Learning: \\ from Manifolds to Cellular Sheaves and Back}
\author{Claudio Battiloro \quad  Zhiyang Wang \quad Hans Riess \quad  Paolo Di Lorenzo \quad Alejandro Ribeiro
\thanks{CB is with the Biostatistics Department, Harvard University, Boston, USA, and with the ESE Department, University of Pennsylvania, Philadelphia, USA, email: cbattiloro@hsph.harvard.edu. ZW and AR are with the ESE Department, University of Pennsylvania, Philadelphia, USA, email: zhiyangw@seas.upenn.edu. HR is with the ECE Department, Duke University, Durham, USA. PDL is with the DIET Department, Sapienza University of Rome, Rome, Italy.  This work was funded by NSF CCF 1934960.  Preliminary results presented in \cite{battiloro2022tnnicassp}.
}
}
\begin{document}
%
\maketitle
\begin{abstract}
\black{In this work we introduce a convolution operation over the tangent bundle of Riemann manifolds in terms of exponentials of the Connection Laplacian operator.} We define tangent bundle filters and tangent bundle neural networks (TNNs) based on this convolution operation, which are novel continuous architectures operating on tangent bundle signals, i.e. vector fields over the manifolds. Tangent bundle filters admit a spectral representation that generalizes the ones of scalar manifold filters, graph filters and standard convolutional filters in continuous time. We then introduce a discretization procedure, both in the space and time domains, to make TNNs implementable, showing that their discrete counterpart is a novel principled variant of the very recently introduced sheaf neural networks. We formally prove that this discretized architecture converges to the underlying continuous TNN. \textcolor{black}{Finally, we numerically evaluate the effectiveness of the proposed architecture on various learning tasks, both on synthetic and real data, comparing it against other state-of-the-art and benchmark architectures.}
\end{abstract}
\begin{IEEEkeywords}
Tangent Bundle Signal Processing, Tangent Bundle Neural Networks, Cellular Sheaves, Sheaf Neural Networks, Graph Signal Processing
\end{IEEEkeywords}
\vspace{-.5cm}
\section{Introduction}\label{sec:intro}
During the last few years, the development of deep learning techniques has led to state-of-the-art results in various fields. More and more sophisticated architectures have promoted significant improvements from both theoretical and practical perspectives. Although it is not the only reason, the success of deep learning is \textcolor{black}{in part} due to Convolutional Neural Networks (CNNs) \cite{lecun1998gradient}. CNNs have achieved excellent performances in a wide range of applications, spanning from image recognition \cite{alexnet2012} to speech analysis \cite{hamid2012cnnspeech} while, at the same time, \textcolor{black}{lightening the computational load of feedforward fully-connected neural networks and integrating features in different spatial resolutions with pooling operators}. CNNs are based on shift operators in the space domain that induce desirable properties in the convolutional filters, among which the most relevant one is the property of shift equivariance. CNNs naturally leverage the regular (often metric) structure of the signals they process, such as spatial or temporal structure. However, data defined on irregular (non-Euclidean) domains are pervasive, with applications ranging from detection and recommendation in social networks \cite{aggarwal2020machine}, to resource allocations over wireless networks \cite{wang2022learning}, and point clouds for shape segmentation \cite{xie2020linking}, just to name a few. Structured data is modeled via the more varied mathematical objects, among which graphs and manifolds are notable examples.  For this reason, the notions of shifts in CNNs have been adapted to convolutional architectures on graphs (GNNs) \cite{gama2018convolutional,scarselli2008graph} as well as a plethora of other structures,~e.g.~simplicial complexes \cite{battiloro2022san,bodnar2021weisfeiler,barbarossa2020topological}, cell complexes \cite{battiloro2022can,bodnarcwnet},  \textcolor{black}{homogeneous spaces \cite{cohen2019general}}, \textcolor{black}{order lattices \cite{riessmultidimensional}}, \textcolor{black}{and manifolds \cite{wang2021stability, cohen2019gauge, schonsheck2018parallel}}. In \cite{parada2020algebraic}, a framework for algebraic neural networks has been proposed exploiting commutative algebras. \textcolor{black}{However, none of these studies consider convolutional filtering of vector fields over manifolds.} Therefore, in this work we focus on tangent bundles, \textcolor{black}{manifolds constructed from the tangent spaces of a domain manifold. Tangent bundles are a specialization of vector bundles which are a specialization of sheaves, all three of which, in increasing levels of generality, mathematically characterize both (1) when local data extends globally and (2) topological obstructions thereof. Our present focus is on tangent bundles as they are a tool for describing and processing vector fields, ubiquitous data structures critical} in tasks such as robot navigation and flocking modeling, as well as in climate science \cite{bermejo2009climate} and astrophysics \cite{collier2018magnet}. Moreover, to make the proposed procedures implementable, we formally describe and leverage the link between tangent bundles and orthogonal cellular sheaves (also called discrete vector bundles), a mathematical structure that generalizes connection graphs and matrix-weighted graphs.
\vspace{-.1cm}
\subsection{Related Works}  
\textcolor{black}{The well-known manifold hypothesis \textcolor{black}{\cite{bronstein2017geometric}} states that high-dimensional data examples are sampled from one (or more) low-dimensional (Riemann) manifolds.} This assumption is the fundamental block of manifold learning, a class of methods for non-linear dimensionality reduction. \textcolor{black}{The Laplacian Eigenmap framework is based on the approximation of manifolds by weighted undirected graphs constructed with $k$-nearest neighbors or proximity radius heuristics}, with the key assumption being that a set of sampled points of the manifold is available \cite{belkin2008towards,chung1997spectral,dunson2021spectral}. \textcolor{black}{Formal connections between GNNs and Manifold Neural Networks (MNNs) are established in \cite{wang2022convolution,levie2021transferability}.}  Most of the previous works focused on scalar signals, e.g. one or more scalar values attached to each node of graphs or point of manifolds; however, recent developments \cite{Sharp2019vhm,hansen2019sheafsp,hansen2020opinion, bodnar2022sheafdiff} showed that processing vector data defined on tangent bundles of manifolds or discrete vector bundles comes with a series of benefits. The work in \cite{Sharp2019vhm} introduced a method for computing parallel transport of vector-valued data on a curved manifold by extending a vector field defined over any region to the rest of the manifold via \textcolor{black}{geodesic curves}. The work in \cite{collier2018magnet} presented an algorithm to reconstruct the magnetopause surfaces from tangent vector observations. \textcolor{black}{Pioneering works on sheaf theory can be found in \cite{leray1998selected, serre1955faisceaux, grothendieck1955general}.} \textcolor{black}{Discrete versions of sheaves, called cellular sheaves, were first introduced in \cite{shepard1985cellular} and were later rediscovered in \cite{curry2014sheaves}. In \cite{shepard1985cellular,curry2014sheaves}, these sheaves were first defined over regular cell complexes, hence the term ``cellular'' sheaves. Often, as in this work, cellular sheaves are defined over tamer objects, here graphs.} \textcolor{black}{In \cite{hansen2019sheafsp}, the authors studied the problem of learning cellular sheaves from (assumed) smooth graph signals.} \textcolor{black}{The work in \cite{hansen2020opinion,Hansen2019towardspecsheaf,ghrist2022cellular,riess2022diffusion} introduced a novel class of diffusion dynamics on cellular sheaves as a model for network dynamics.} In \cite{bodnar2022sheafdiff,hansen2020sheafnn,barbero2022sheafnnconn}, neural networks operating on discrete vector bundles are presented, generalizing GNNs: additionally, the work in \cite{bodnar2022sheafdiff} exploited cellular sheaf theory to show that the underlying geometry of the graph \textcolor{black}{gives rise to} oversmoothing behavior of GNNs. Finally, the most important works for us are  \cite{singer2012vdm, singer2017spectral}. In particular, in \cite{singer2012vdm}, \black{the authors introduced an algorithmic generalization of non-linear dimensionality reduction methods based on the Connection Laplacian operator and proved that both manifolds and their tangent bundles can be approximated with certain cellular sheaves constructed from sampled points of the manifolds.} The work in \cite{singer2017spectral} further generalized the result of \cite{singer2012vdm} by presenting a framework for approximating  Connection Laplacians over manifolds via their principal bundle structure, and by proving \black{the spectral convergence of the approximating sheaf Laplacians.}
\vspace{-.5cm}
\subsection{Contributions.} In this work, we first define a \textit{convolution operation over the tangent bundle} of Riemann manifolds via the Connection Laplacian operator. Our definition is derived from the vector diffusion equation over manifolds, \textcolor{black}{and generalizes convolutions on manifolds \cite{wang2022convolution}, graphs \cite{shuman2013emerging, gama2018convolutional}, as well as standard time convolutions}. Leveraging this operation, we introduce \textit{Tangent Bundle Convolutional Filters} to process tangent bundle signals \textcolor{black}{(vector fields)}. We define the \textit{frequency representation} of tangent bundle signals and the \textit{frequency response} of tangent bundle filters using the spectral properties of the Connection Laplacian. By cascading layers consisting of tangent bundle filter banks and pointwise non-linearities, we introduce \textit{Tangent Bundle Neural Networks} (TNNs). \textcolor{black}{The proposed convolutional processing framework can be also seen as a novel instantiation of the general theory of algebraic signal processing  \cite{parada2020algebraic,puschel2008algebraic}.} However, tangent bundle filters and tangent bundle neural networks are continuous architectures that cannot be directly implemented in practice. \textcolor{black}{For this reason, we provide a principled way of discretizing them, both in time and space domains, making convolutions on them computable}. In particular, we discretize the TNNs in the space domain by sampling points on the manifold and building a cellular sheaf \cite{Hansen2019towardspecsheaf} that represents a \textcolor{black}{legitimate} approximation of both the manifold and its tangent bundle \cite{singer2012vdm}. We \textit{prove that the space discretized architecture over the cellular sheaf converges to the underlying TNN} as the number of sampled points increases. Moreover, we further discretize the architecture in the time domain by sampling the filter impulse function in discrete and finite time steps, notably showing that space-time discretized TNNs (DD-TNNs) are a novel principled variant of the very recently introduced Sheaf Neural Networks \cite{bodnar2022sheafdiff,hansen2020sheafnn,barbero2022sheafnnconn}, and thus shedding further light, from a theoretical point of view, on the deep connection between algebraic topology and differential geometry. Finally, we evaluate the performance of TNNs on both synthetic and real data; in particular, we design a denoising task of a synthetic tangent vector field on the torus, \textcolor{black}{a manifold classification task}, a reconstruction task, and a forecasting task of the daily Earth wind field, tackled via a recurrent version of our architecture. We empirically demonstrate the advantage of incorporating the tangent bundle structure into our model by comparing TNNs against Manifold Neural Networks from \cite{wang2022convolution} \textcolor{black}{(architectures taking into account the manifold structure, but not the tangent spaces),} \textcolor{black}{Multi-Layer Perceptrons \cite{haykin1994neural}, and Recurrent Neural Networks (the latter two do not consider any geometric information).}
\vspace{-.1cm}
\subsection{Paper Outline} The paper is organized as follows. We introduce some preliminary concepts in Section \ref{subsec:prelim_def}.  We define tangent bundle convolution, filters and neural networks in Section \ref{sec:filters}. In Section \ref{sec:disc}, we illustrate the proposed discretization procedure for TNNs and we prove the convergence result. \textcolor{black}{We discuss the consistency of the proposed convolution in Section \ref{sec:consistency}}.Numerical results are in Section \ref{sec:num_res}, and conclusions in Section \ref{sec:conlcusion}. 
\begin{table}[h]
    \centering
    \textcolor{black}{\scalebox{.75}{\begin{tabular}{|c|c|}
    \hline
        Manifold & $\mathcal{M}$\\
        \hline\xrowht[()]{4pt}
        Tangent Space at point $x$ & $\tmx$ \\
        \hline\xrowht[()]{4pt}
   Tangent Bundle & $\tm$ \\
        \hline\xrowht[()]{4pt}
        Tangent Bundle Signal & $\F: \M \rightarrow \tm$ \\
        \hline\xrowht[()]{4pt}
    Differential & $d\i:\tmx\rightarrow\tmxrp$\\
        \hline\xrowht[()]{4pt}
    Riemann Metric & $\langle~,~\rangle_{\tmx} : \tmx \times \tmx \rightarrow \mathbb{R}$
    \tabularnewline
    \hline
    \end{tabular}}}
    \caption{\textcolor{black}{Notation}}
    \label{not_Table}
\end{table}
\vspace{-.5cm}
\section{Preliminary Definitions}\label{subsec:prelim_def}
In this section, we review some concepts from Riemann geometry that will be useful to introduce the convolution operation over tangent bundles.
\begin{figure}[t]
    \centering
    \includegraphics[scale = .34]{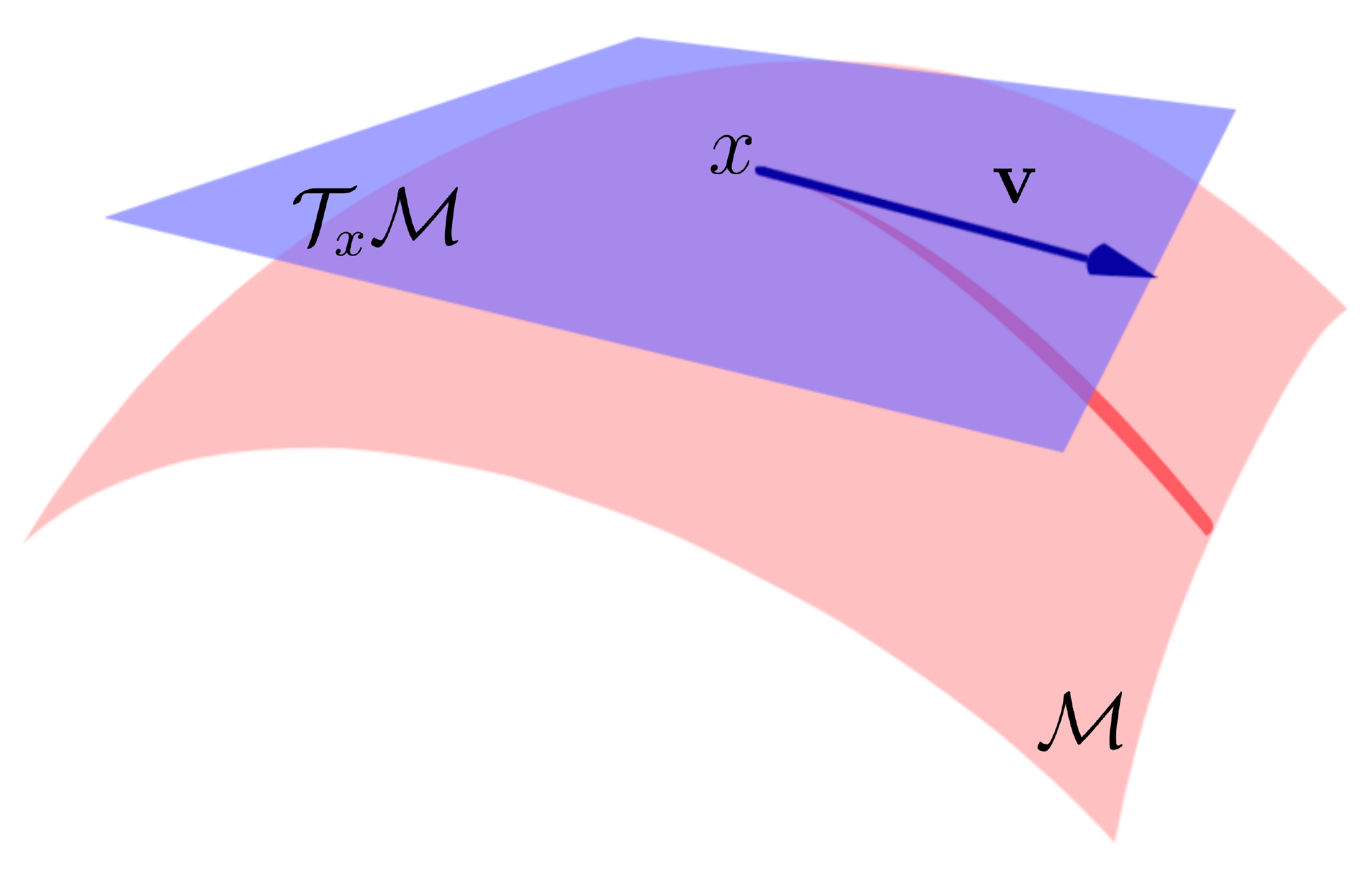}
    \caption{An example of tangent vector}
    \label{fig:tangent_space}
\end{figure}
\vspace{-.1cm}
\subsection{Manifolds and Tangent Bundles} 
We consider a compact, smooth, and \textcolor{black}{orientable} $d$-dimensional manifold $\mathcal{M}$ \textcolor{black}{smoothly} embedded in $\mathbb{R}^p$. Each point $x \in \mathcal{M}$ is endowed with a $d$-dimensional tangent  space \textcolor{black}{$\tmx$ isomorphic to  $\mathbb{R}^d$}, whose elements $\mathbf{v} \in \tmx$ are said to be tangent vectors at $x$. \textcolor{black}{For explicit construction of tangent spaces on a manifold, consult an introductory textbook on differential topology \cite{lee2006Riemannian}. Informally, tangent vectors can be seen as a generalization of the velocity vector of a curve constrained to $\mathcal{M}$ passing through the point $x$. An example of a tangent vector is depicted in Fig. \ref{fig:tangent_space}}.
\textcolor{black}{
\begin{definition}[Tangent Bundle]
    The tangent bundle is the disjoint union of the tangent spaces $\tm = \bigsqcup_{x \in \M} \tmx$ together with the projection map $\pi: \tm \to \M$ given by $\pi(x,\mathbf{v}) = x$.
\end{definition}}
\textcolor{black}{Moreover, the tangent bundle has a natural topology which makes it a smooth $2d$-manifold and makes $\pi$ a smooth map \cite{lee00manif}.} In abuse of language, we often refer to the tangent bundle as simply the space $\tm$. The embedding induces a Riemann structure on $\mathcal{M}$ which allows to equip each tangent space $\tmx$ with an inner product.
\textcolor{black}{\begin{definition}[Riemann Metric]
    A Riemann Metric on a compact and smooth $d$-dimensional manifold $\mathcal{M}$ embedded in $\mathbb{R}^p$ is a (smoothly chosen) inner product $\langle~,~\rangle_{\tmx} : \tmx \times \tmx \rightarrow \mathbb{R}$ on each of the tangent
spaces $\tmx$ of $\mathcal{M}$ given, for each $\mathbf{v}$,$\mathbf{w} \in \tmx$, by
\begin{equation}
\label{riemann_metric}
\langle\mathbf{v},\mathbf{w}\rangle_{\tmx} = \langle d\i \mathbf{v},  d\i\mathbf{w} \rangle_{\mathbb{R}^p},
\end{equation}
where $d\i\mathbf{v} \in \tmxrp$ is \textcolor{black}{called the} differential of $\mathbf{v} \in \tmx$ in $\tmxrp \subset \rp$, $\tmxrp$ is the $d$-dimensional subspace of $\rp$ being the embedding of $\tmx$ in $\rp$,  \textcolor{black}{the differential $d\i:\tmx\rightarrow\tmxrp$ is an injective linear mapping} \textcolor{black}{(also referred to as pushforward, as it pushes tangent vectors on $\mathcal{M}$ forward to tangent vectors on $\mathbb{R}^p$)} \cite{lee2006Riemannian}, and $\langle,\rangle_{\mathbb{R}^p}$ is the usual dot product.
\end{definition}}
The  Riemann metric induces also a uniform probability measure $\mu$ over the manifold, simply given by the considered region scaled by the volume of the manifold.
\vspace{-.1cm}
\subsection{Tangent Bundle Signals} 
A tangent bundle signal is a vector field over the manifold, thus a mapping $\F: \M \rightarrow \tm$ that associates to each point of the manifold a vector in the corresponding tangent space. \textcolor{black}{In the theory of vector bundles, a bundle signal is a section.} An example of a (sparse) tangent vector field over the unit $2$-sphere is depicted in Fig. \ref{fig:sphere}\cite{battiloro2022tnnicassp}. 
\textcolor{black}{\begin{remark}
    The choice of employing the terminology "tangent bundle signal" and not the standard "vector fields" or ``sections'' aims to further underline the strong signal processing perspective of this work, and to facilitate the understanding of its generalization properties, as highlighted in Section \ref{sec:consistency}.
\end{remark}}
We can define an inner product for tangent bundle signals in the following way.
\begin{definition}[Tangent Bundle Inner Product]
    Given tangent bundle signals $\F$ and $\mathbf{G}$, their inner product is given by
\begin{align}
\label{inn_prod}
\langle \F, \mathbf{G} \rangle_{\tm}
&= \int_{\M} \langle \F(x), \mathbf{G}(x) \rangle_{\tmx} \textrm{d}\mu(x),
\end{align}
and the induced norm is $||\F||^2_{\tm} = \langle \F, \F \rangle_{\tm}$.
\end{definition}
\textcolor{black}{We denote with $\ltm$ the space of tangent bundle signals. Note that tangent bundle signals have finite energy with respect to $|| \cdot ||_{\tm}$, because they are (continuous) sections of the tangent bundle. Therefore, the length of all the vectors in a vector field is bounded because the image of a continuous function on a compact set is bounded. Hence, integrating a bounded function on (compact) $\mathcal{M}$ is always well-defined.} In the following, we denote $\langle \cdot, \cdot \rangle_{\tm}$ with $\langle \cdot, \cdot \rangle$  when there is no risk of confusion.
\vspace{-.1cm}
\section{Tangent Bundle Convolutional Filters}\label{sec:filters}
\textcolor{black}{Linear filtering operations are historically synonymous (under appropriate assumptions) with convolution.} Time signals are filtered by computing the
continuous-time convolution of the input signal and the filter
impulse response [17]; images  are filtered by computing multidimensional
convolutions [34]; graph signals are filtered by computing
graph convolutions [5]; scalar manifold signals are filtered by computing manifold convolutions \cite{wang2022convolution}. In this paper, we define a tangent bundle filter as the convolution of the filter impulse response $\th$ and the tangent bundle signal $\F$. To do so, we exploit the Connection Laplacian Operator.
\begin{figure}
    \centering
    \includegraphics[scale = .068]{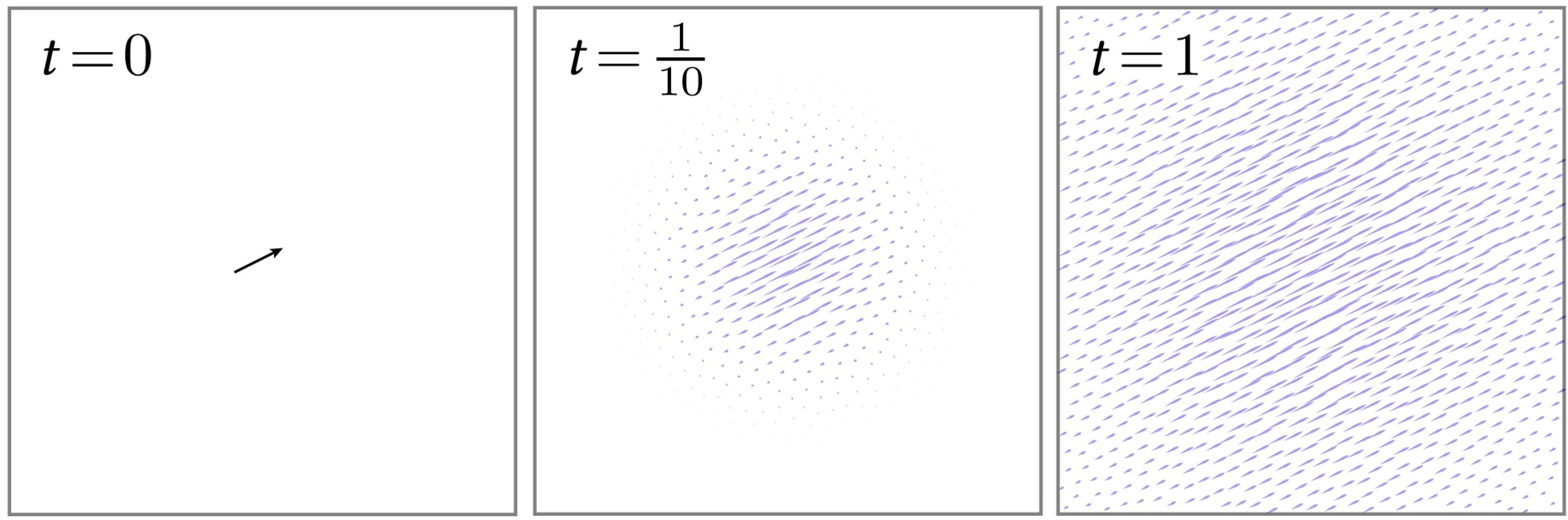}
    \caption{\textcolor{black}{Vector diffusion}}
    \label{fig:vec_diff}
\end{figure}
\vspace{-.1cm}
\subsection{Connection Laplacian}
 The Connection Laplacian is a (second-order) operator $\Delta: \ltm \rightarrow \ltm$, given by the trace of the second covariant derivative defined (for this work) via the Levi-Civita connection \cite{singer2012vdm} \textcolor{black}{(the unique connection compatible with the Riemann metric).}  The Connection Laplacian $\Delta$ has some desirable properties: it is negative semidefinite, self-adjoint, elliptic, and, furthermore, 
 has a negative  spectrum $\{-\lambda_i, \Phii\}_{i=1}^{\infty}$ with eigenvalues $\lambda_i$ and  corresponding eigenvector fields $\Phii$ satisfying
\begin{equation}
\label{eigen}
\Delta\Phii = -\lambda_i\Phii,
\end{equation}
with $0<\lambda_1\leq\lambda_2\leq\cdots$. The only possible accumulation \textcolor{black}{(limit)} point is $-\infty$ \black{\cite{singer2012vdm}} We can use the Connection Laplacian \textcolor{black}{to fathom a heat equation for vector diffusion:}
\begin{equation}
\label{diff_eq}
\frac{\partial \U(x,t)}{\partial t} - \Delta\U(x,t) = 0,
\end{equation}
where $\U: \M \times \mathbb{R}_0^+ \rightarrow \tm$ and $\U( \cdot, t) \in \ltm \, \forall t \in \mathbb{R}_0^+$; we denote the initial condition condition with $\U( x, 0) = \F(x)$. \textcolor{black}{As reported in \cite{Sharp2019vhm} and in Fig. \ref{fig:vec_diff} (obtained from Fig. 4 of \cite{Sharp2019vhm}), an intuitive interpretation of  \eqref{diff_eq} is imagining the evolution of the vector field $\U(x,t)$ over time as a "smearing out" of the initial vector field $\F(x)$}. \textcolor{black}{In this interpretation, the role of the Connection Laplacian can be understood as a means to diffuse vectors from one tangent space to another, because it encodes when tangent vectors are parallel (via the connection), and how to "move" them keeping them parallel (via the induced parallel transport). On scalar functions on Euclidean domains, it agrees with the classical Laplace operator. (Indeed, in the flat case it is sufficient to independently diffuse each scalar component, however, this approach fails for curved space.)} The solution of \eqref{diff_eq} is given by 
\begin{equation}
\label{exp_sol}
\U( x, t) = e^{t\Delta}\F(x),
\end{equation}
which provides a way to construct tangent bundle convolution, as explained in the following section. 
\begin{figure}[t]
    \centering
    \includegraphics[scale = .31]{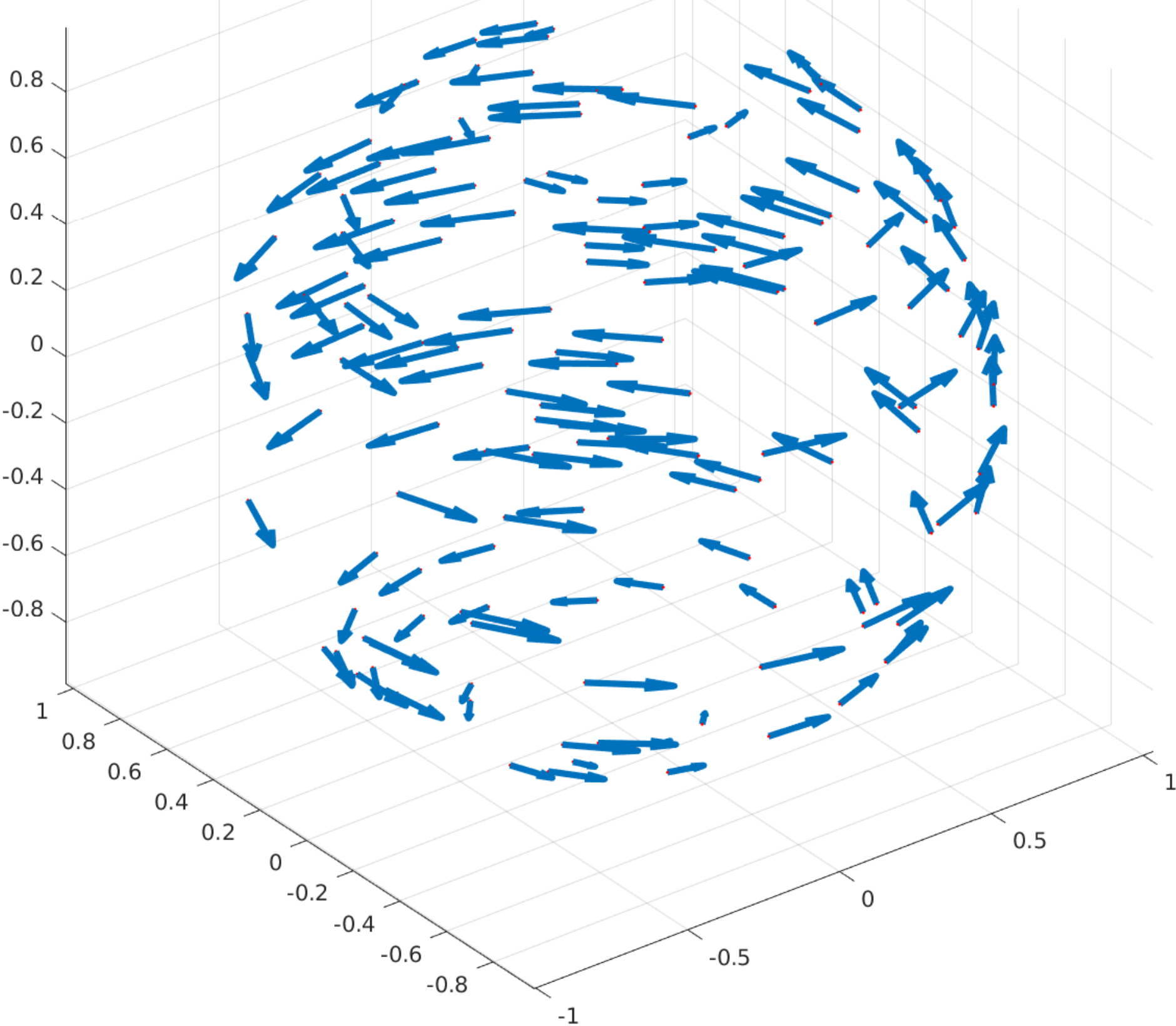}
    \caption{An example of tangent bundle signal}
    \label{fig:sphere}
\end{figure}
\vspace{-.1cm}
\subsection{Tangent Bundle Filters}
We are now in the condition of defining a convolution operation and tangent bundle convolutional filters leveraging the heat diffusion dynamics in \eqref{diff_eq}. 
\textcolor{black}{\begin{definition}[Tangent Bundle Filter] \label{def:tangent-bundle-filter}
Let $\th:\mathbb{R}^+ \rightarrow \mathbb{R}$ and let $\F \in \ltm$ be a tangent bundle signal. The tangent bundle filter with impulse response $\th$, denoted with $\h$, is given by
\begin{equation}
    \label{convolution}
    \G(x) = \big(\th \star_{\tm} \F\big)= \int_0^{\infty}\th(t)\U(x,t)\textrm{d}t,
\end{equation}
where $\star_{\tm}$ is the \textit{tangent bundle convolution}, and $\U(x,t)$ is the solution of the heat equation in \eqref{diff_eq} with $\U(x,0) = \F(x)$.
\end{definition}}
In the following, we will use the terms tangent bundle filter and tangent bundle convolution interchangeably. \textcolor{black}{One cannot explicity compute the output $\G$ directly from the input $\F$ in  Definition \ref{def:tangent-bundle-filter}. However,
this is remedied by injecting the solution of the heat equation \eqref{exp_sol} into
\eqref{convolution}.} In this way, we can derive a closed-form expression for $\h$ that is parametric on the Connection Laplacian, as shown in the following proposition.
\begin{proposition}[Parametric Filter] \label{prop:parametric-filter}
    Any tangent bundle filter $\h$ defined as in \eqref{convolution} is a parametric map $\h(\Delta)$ of the Connection Laplacian operator $\Delta$, given by
\begin{equation}
    \label{param_conv}
    \G(x) = \h\F(x) =  \int_0^{\infty}\th(t)e^{t\Delta}\F(x)\textrm{d}t = \h(\Delta)\F(x).
\end{equation}
\end{proposition}
We can make several considerations starting from Proposition \ref{prop:parametric-filter}: we can state that tangent bundle filters are spatial operators,
since they operate directly on points $x \in \mathcal{M}$; moreover,  they are local operators, because they are parametrized by $\Delta$ which is itself a local operator.
\begin{remark}
The exponential term $e^{t\Delta}$ can be seen as a diffusion or shift operator similar to a time delay in a linear time-invariant (LTI) filter \cite{oppenheim1997signals}, or to a graph shift operator in a linear shift-invariant (LSI) graph filter \cite{gama2020gcnmagazine}, or to a manifold shift operator based on the Laplace-Beltrami operator \cite{wang2022convolution}. The resemblance is due to the fact that tangent bundle filters are linear combinations of the elements of the tangent bundle diffusion sequence, such as graph filters are linear combinations of
the elements of the graph diffusion sequence and manifold filters are linear combinations of the elements of the manifold diffusion sequence. These considerations are further useful to validate the consistency of the proposed convolution operation, discussed in detail in Section \ref{sec:consistency}.
\end{remark}
\vspace{-.1cm}
\subsection{Frequency Representation of Tangent Bundles Filters}
The spectral properties of the Connection Laplacian $\Delta$ allow us to introduce the \textcolor{black}{notion of a frequency domain. Following the approach historically common to many signal processing frameworks, we define the frequency representation of a tangent bundle signal $\F \in \Gamma(\tm)$ as its projection onto the eigenbasis of the Connection Laplacian}
\begin{equation}
\label{freq_resp}
    \big[\hat{F}\big]_i = \langle \F, \Phii \rangle = \int_{\M}\langle \F(x), \Phii(x) \rangle_{\tmx} \textrm{d}\mu(x).
\end{equation}
\begin{proposition}[Frequency Representation] \label{prop:frequencey-rep}
    Given a tangent bundle signal $\F$ and a tangent bundle filter $\h(\Delta)$ as in Definition \ref{def:tangent-bundle-filter}, the frequency representation of the filtered signal $\G = \h(\Delta)\F$ is given by
\begin{equation}
    \big[\hat{G}\big]_i = \int_0^{\infty} \th(t)e^{-t\lambda_i}\textrm{d}t\big[\hat{F}\big]_i.
\end{equation}
\end{proposition}
\begin{proof}
    See Section B of Supplemental Material.
\end{proof}
Therefore, we can characterize the frequency response of a tangent bundle filter in the following way.
\begin{definition}[Frequency Response] \label{def:frequency}
    The frequency response $\hat{h}(\lambda)$ of the filter $\h(\Delta)$ is defined as
\begin{equation}
    \label{frequency_filt}
    \hat{h}(\lambda) = \int_0^{\infty} \th(t)e^{-t\lambda}\textrm{d}t.
\end{equation}
\end{definition}
This leads to $\big[\hat{G}\big]_i = \hat{h}(\lambda_i)\big[\hat{F}\big]_i$, meaning that the tangent bundle filter is point-wise in the frequency domain. We can finally write the frequency representation of the filter as
\begin{equation}
\label{filtered_series}
\G = \h(\Delta)\F = \sum_{i=1}^{\infty} \hat{h}(\lambda_i) \langle \F, \Phii \rangle \Phii.
\end{equation}
{\color{black}\begin{remark}
    The frequency response $\hat{h}(\lambda)$ in Definition \ref{def:frequency} is the Laplace transform of $\th(t)$ if we let $\lambda$ be an arbitrary complex argument. The effect of a tangent bundle filter on a tangent bundle signal in the frequency domain is determined by evaluating $\hat{h}(\lambda)$ at the eigenvalues $\lam_i$ of the Connection Laplacian. This interpretation is analogous to the interpretation of the Fourier transform as an instantiation of the Laplace transform restricted to $\lambda = j\omega$ \cite{oppenheim1997signals}. This analogy can be furthered by observing that $j\omega$ are eigenvalues of the derivative operator (see Section \ref{sec:consistency}). This interpretation is also consistent with the interpretation of the frequency response of manifold filters---also a Laplace transform which is instantiated at the eigenvalues of the Laplace-Beltrami operator \cite{wang2022convolution}---and the frequency response of graph filters -- a $z$-transform which is instantiated at the eigenvalues of the graph shift operator \cite{shuman2013emerging}.
\end{remark}}
\vspace{-.1cm}
\subsection{Lowpass Tangent Bundle Filters}
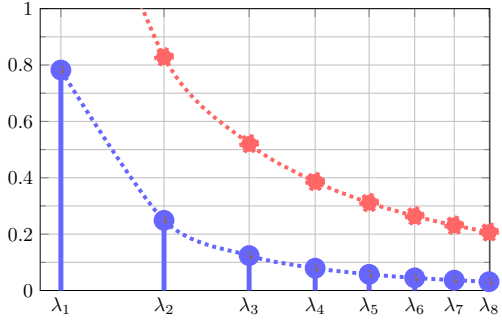
\begin{figure}[t!]
    \centering
    \scalebox{.75}{\begin{tikzpicture}
\begin{axis}[
    xlabel={},
    ylabel={},
    ymin=0, ymax=1, 
    xmin=0.68, xmax=2.2,
    scale only axis,
     xtick = { 0., 0.75,  1.09861229,  1.38629436, 1.60943791, 1.79175947, 1.94591015, 2.07944154, 2.19722458},
     xticklabels = {$\lambda_1$,
     $\lambda_1$,
     $\lambda_2$,
     $\lambda_3$,
     $\lambda_4$,
     $\lambda_5$,
     $\lambda_6$,
     $\lambda_7$,
     $\lambda_8$,},
    samples at={0., 0.75, 1.09861229, 1.38629436, 1.60943791,
       1.79175947, 1.94591015, 2.07944154, 2.19722458},
    grid=both,
    minor tick num=1,
    every axis plot/.append style={thick},
    height = 5cm,
    width = .9\linewidth
]

\addplot+[ycomb, mark=*, blue!60, mark options={fill=blue!60, scale=2}, line width   = 2.5] plot (x, {.33/(x^3});
\addplot+[smooth, mark=*, red!60, mark options={fill=red!60, scale=2}, line width   = 2, dotted] expression {1/(x^2};

\addplot+[smooth, blue!60, line width   = 2, dotted] expression {.33/(x^3};

\end{axis}
\end{tikzpicture}}
    \caption{\textcolor{black}{Illustration of a lowpass, non-amplifying, Lipschitz continuous tangent bundle filter. The $x$-axis stands for the spectrum with each sample representing an eigenvalue. Here the eigenvalues increase at a logarithmic rate. The red dotted line is $\lambda_i^{-2}$ and  the blue dotted line is the filter, obtained with impulse response $\widetilde{h}(t) = t^2/6$, thus $\hhath(\lambda) = \frac{1}{3}\lambda^{-3}$, from \eqref{frequency_filt}.}}
    \label{fig:alpha_fdt}
\end{figure}
The spectrum of the Connection Laplacian $\Delta$ is infinite-dimensional, i.e., there is an infinite (though countable) number of eigenvalues that need to be taken into account. \textcolor{black}{However, we can design lowpass filters to tackle this problem. This design, although not mandatory for practical purposes, is crucial in proving the convergence result of the discretized filters and neural networks to the underlying continuous filters and TNNs, respectively, stated in Theorem \ref{theorem-main}. 
\begin{definition}[Lowpass Tangent Bundle Filters]\label{def:lowpass} A tangent bundle filter $\mathbf{h}(\Delta)$ is a lowpass filter if its frequency response function $\hhath$ is $\mathcal{O}(\lambda_i^{-2})$, i.e. if $\limsup_{i \rightarrow \infty}\hhath(\lambda_i)\lambda_i^2 < \infty$.
\end{definition}}
In other words, lowpass filters asymptotically decay at least as fast as $\lambda_i^2$, thus progressively suppressing high frequencies. Finally, we define Lipshitz continuous and non-amplifying tangent bundle filters.
\begin{definition}[Tangent Bundle Filters with Lipschitz Continuity]
    A tangent bundle filter is $C$-Lispchitz if its frequency response is Lipschitz continuous with constant $C$, i.e if 
     $|\hat{h}(a)-\hat{h}(b)| \leq C |a-b|\text{ for all } a,b\in (0,\infty)\text{.}$
\end{definition}
\begin{definition}[Non-Amplifying Tangent Bundle Filters]
    A tangent bundle filter is non-amplifying if for all $\lambda\in(0,\infty)$, its frequency response $\hat{h}$ satisfies $|\hat{h}(\lambda)|\leq 1$.
\end{definition}
The Lipschitz continuity is a standard assumption, while the non-amplifying assumption is perfectly reasonable, as any (finite-energy) filter function $\hhath(\lambda)$ can be normalized.  An example of a lowpass, non-amplifying, Lipschitz continuous tangent bundle filter is depicted in Fig.~\ref{fig:alpha_fdt}.
\vspace{-.1cm}
\subsection{Tangent Bundle Neural Networks}\label{sec:tnn}

We define a layer of a Tangent Bundle Neural Network (TNN) as a bank of tangent bundle filters followed by a pointwise non-linearity. In this setting,  pointwise informally means ``pointwise in the ambient space''. We introduce the notion of differential-preserving non-linearity to formalize this concept in a consistent way.
\begin{definition}[Differential-preserving Non-Linearity] \label{def:pointwise}
    \textcolor{black}{Denote with $U_x \subset \tmxrp$ the image of the injective differential $d\i$ in $\tmxrp$. A mapping $\sigma:\ltm \rightarrow \ltm$ is a differential-preserving non-linearity if it can be written as $\sigma(\F(x)) = d\i^{-1} \widetilde{\sigma}_x  (d\i  \F(x))$, where $\widetilde{\sigma}_x: U _x\rightarrow U_x$ is a point-wise non-linearity in the usual (Euclidean) sense.}
\end{definition}
Furthermore, we assume that $\widetilde{\sigma}_x = \widetilde{\sigma}$ for all $x \in \M$. 
\begin{definition}[Tangent Bundle Neural Networks]
    The $l$-th layer of a TNN with $F_l$ input signals $\{\F_l^q\}_{q = 1}^{F_l}$, $F_{l+1}$ output signals $\{\F_{l+1}^u\}_{u = 1}^{F_{l+1}}$, and non-linearity $\sigma(\cdot)$ is defined as
\begin{equation}
    \label{tnn_layer}
    \F_{l+1}^u(x) = \sigma\Bigg(\sum_{q=1}^{F_l}\h(\Delta)_l^{u,q}\F_l^q(x)\Bigg), \; u = 1,...,F_{l+1}.
\end{equation}
\end{definition}
Therefore, a  TNN of depth $L$ with input signals $\{\mathbf{F}^q\}_{q=1}^{F_0}$ is built as the stack of $L$ layers defined as in \eqref{tnn_layer}, where $\F_0^q = \F^q$. An additional task-dependent readout layer \textcolor{black}{(e.g~sum for classification) can be appended to the final layer.}

To globally represent the TNN, we collect all the filter impulse responses in a function set $\mathcal{H} = \big\{\widetilde{h}_l^{u,q}\big\}_{l,u,q}$ and we describe the TNN $u$-th output as a mapping $\F^u_{L}=\bPsi_u\big(\mathcal{H}, \Delta, \{\mathbf{F}^q\}_{q=1}^{F_0}\big)$ to emphasize that at TNN is parameterized by both $\mathcal{H}$ and the Connection Laplacian $\Delta$.
\vspace{-.2cm}
\section{Discretization in Space and Time} \label{sec:disc}
\vspace{-.1cm}
Tangent Bundle Filters and Tangent Bundle Neural Networks operate on tangent bundle signals, thus they are continuous architectures that cannot be directly implemented in practice. Here we provide a procedure for discretizing tangent bundle signals, both in time and spatial domains; the discretized counterpart of TNNs \textcolor{black}{ is an instantiation of the recently introduced Sheaf Neural Networks \cite{hansen2020sheafnn}.} For this reason, in this section we first provide a brief review of cellular sheaves over undirected graphs, and then we explain the proposed discretization procedure.
\vspace{-.1cm}
\subsection{Cellular Sheaves}
\textcolor{black}{A cellular sheaf over an (undirected) graph consists of a vector space for each node and edge and a collection of linear transformations indexed by node-edge incidence pairs of the graph. Formally, it is a functor on a partially ordered set of node-edge incidence relations into the category of vector spaces and linear transformations. We introduce the following non-standard notation to emphasize the role that sheaves play in approximating tangent bundles as the number of nodes increases.
\begin{definition}[Cellular Sheaf over a Graph]
    Suppose $\mathcal{M}_n = (\mathcal{V}_n, \mathcal{E}_n)$ is an undirected graph with $n = |\mathcal{V}_n|$ nodes. A cellular sheaf over $\mathcal{M}_n$ is the tuple $\tm_n = (\mathcal{M}_n, \mathcal{F})$, i.e.:
\begin{itemize}
    \item A vector space $\mathcal{F}(v)$ for each $v \in \mathcal{V}_n$. We refer to these vector spaces as node stalks.
    \item A vector space $\mathcal{F}(e)$ for each $e \in \mathcal{E}_n$. We refer to these vector spaces as edge stalks.
    \item A linear mapping $V_{v,e} : \mathcal{F}(v) \rightarrow \mathcal{F}(e)$ represented by a matrix $\mathbf{V}_{v,e}$ for each pair $(v,e) \in \mathcal{V}_n \times \mathcal{E}_n$ with incidence $v~\inc~e$. These mappings are called restriction maps.
\end{itemize}
The space $\mathcal{L}^2(\tm_n) = \bigoplus_{v \in \mathcal{V}} \mathcal{F}(v)$ formed by the direct sum of vector spaces associated with the nodes of the graph is commonly called the space of $0$-cochains, which we refer to as sheaf signals on $\tm_n$. We write a sheaf signal on $\mathcal{M}_n$ as $\mathbf{f}_n \in \mathcal{L}^2(\tm_n)$.
\end{definition}
\begin{definition}[Sheaf Laplacian]
The (non-normalized) Sheaf Laplacian of a sheaf $\tm_n$ is a linear mapping $\Delta_n:\mathcal{L}^2(\tm_n) \rightarrow \mathcal{L}^2(\tm_n)$ defined node-wise
\begin{equation}
    (\Delta_n\mathbf{f}_n)(v) = \sum_{v \inc e \coinc u}\mathbf{V}_{v,e}^T(\mathbf{V}_{v,e}\mathbf{f}_n(v) - \mathbf{V}_{u,e}\mathbf{f}_n(u)).
\end{equation}
\end{definition}
While in general, the dimensions of the stalks may be arbitrary, this work focuses on discrete $\mathcal{O}(d)$-bundles, or orthogonal sheaves. In an orthogonal sheaf, we have $\mathbf{V}_{v,e}^{-1} = \mathbf{V}_{v,e}^{T}$ for all $v \inc e$ and $\mathcal{F}(v) \cong \mathbb{R}^d$ for all $v$.} \textcolor{black}{Note, that this does not mean every stalk is equal, but has the same dimension.} 
\begin{figure}
    \centering
    \includegraphics[scale = .28]{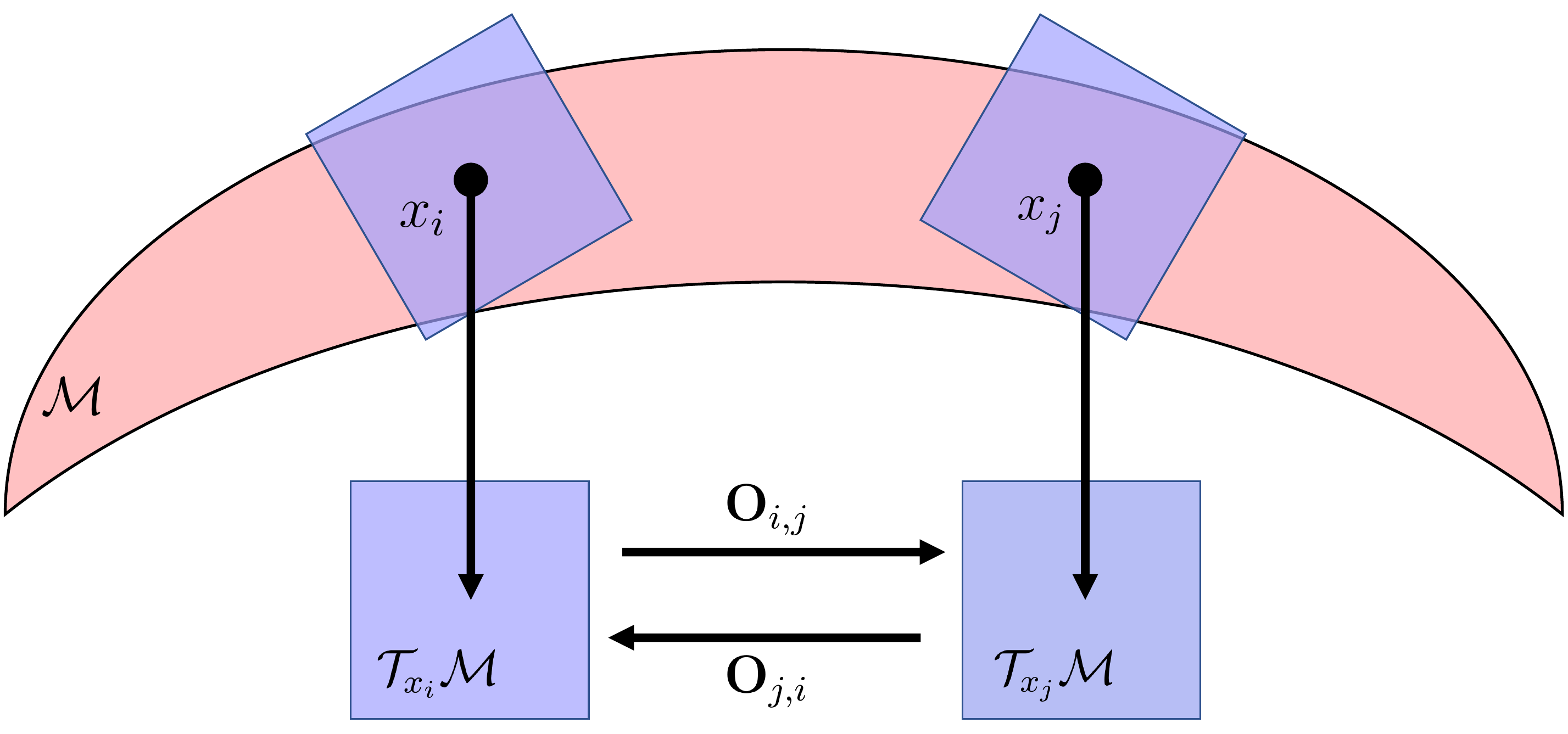} 
    \caption{Pictorial view of discrete parallel transport.}
    \label{fig:transport}
\end{figure}
\vspace{-.2cm}
\subsection{Discretization in the Space Domain}\label{sec:space_disc}
The manifold $\M$, the tangent bundle $\tm$, and the Connection Laplacian $\Delta$ can be approximated from a set of sampled points $\mathcal{X} \subset \mathbb{R}^p$. Knowing the coordinates of the sampled points, we construct an orthogonal cellular sheaf over an undirected geometric graph such that its normalized Sheaf Laplacian converges to the manifold Connection Laplacian as the number of sampled points (nodes) increases \cite{singer2017spectral}. Formally,  we assume that a set of $n$  points $\mathcal{X}=\{x_1,\dots,x_n\}\subset \mathbb{R}^p$ are sampled i.i.d. from measure $\mu$ over $\mathcal{M}$. We build a cellular sheaf $\tm_n$ via the Vector Diffusion Maps procedure whose details are listed in \cite{singer2012vdm} and which we briefly review here.

We start by building a weighted (geometric) graph $\M_n = (\mathcal{V}_n, \mathcal{E}_n$) with nodes $\mathcal{V}_n = \{1,2,\dots, n \}$ and weights $w_{ij}$ for nodes $i$ and $j$ as follows. Set a scale $\epsilon_n>0$. For each pair $i,j \in \mathcal{V}_n \times \mathcal{V}_n$, if $\|x_i - x_j \|_2^2 \leq \epsilon_n$, then $ij \in \mathcal{E}_n$ with weight
\begin{equation}
\label{graph_weights} 
    w_{i,j} = \exp\Bigg(\frac{||x_i-x_j||_{2}}{\sqrt{\epsilon_n}}\Bigg);
\end{equation}
otherwise, $ij \notin \mathcal{E}_n$ and $w_{i,j} = 0$ \textcolor{black}{\cite{singer2012vdm} (Eq. 2.5, page 6)}.
The neighborhood $\mathcal{N}_i$ of each point $x_i$ contains the points $x_j \in \mathcal{X}$ lying in a ball of radius $\sqrt{\epsilon_n}$ centered at $x_i$. Using a local PCA procedure, we assign to each node $i$ an orthogonal transformation $\Oi \in \mathbb{R}^{p\times\hd}$, that is an approximation of a basis of the tangent space $\mathcal{T}_{x_i}\mathcal{M}$, with $\hd$ being an estimate of $d$ obtained from the same procedure (or $d$ itself, if known). In particular, we fix another scale parameter $\epsilon_{\textrm{PCA}}$ (different from the graph kernel scale parameter $\epsilon_n$) and we define the PCA neighborhood $\mathcal{N}^{\textrm{P}}_i$ of each point $x_i$ as the points $x_j \in \mathcal{X}$ lying in a ball of radius $\sqrt{\epsilon_{\textrm{PCA}}}$ centered at $x_i$. We define $\mathbf{X}_i \in \mathbb{R}^{p \times |\mathcal{N}^{\mathrm{P}}_i|}$ for each point to be a matrix whose $j$-th column is the vector $x_j - x_i$, with $x_j \in \mathcal{N}^{\mathrm{P}}_i$; equivalently, it is possible to shift each neighbor by the mean $1/|\mathcal{N}^{\textrm{P}}_i|\sum_{x_j \in \mathcal{N}^{\mathrm{P}}_i} x_j$. 
At this point, we compute for each point a matrix $\mathbf{B}_i = \mathbf{X}_i\mathbf{C}_i$, where $\mathbf{C}_i$ is a diagonal matrix whose entry are defined as $\mathbf{C}(i,i) = \sqrt{K(||x_i -x_j||_{2}/\sqrt{\epsilon_{\textrm{PCA}}})}$, with $K(\cdot)$ being any twice differentiable positive monotone function supported on $[0,1]$ \textcolor{black}{(this scaling is useful to emphasize nearby points over far away points). We now perform the actual Local PCA by computing, per each point, the following covariance matrix and its eigendecomposition
\begin{equation}\label{cov_mat}
\mathbf{R}_i=\B_i^T\B_i=\mathbf{M}_i\Sigma_i\mathbf{M}_i^T.
\end{equation}}
\vspace{-.4cm}
\begin{definition}[Approximated Tangent Space \cite{singer2012vdm} \textcolor{black}{(Eq. 2.1, page 5)}] \label{def:approxtangent}
    For each point $x_i \in \mathcal{X}\subset \mathcal{M}$, the approximated basis $\mathbf{O}_i$ of its tangent space $\mathcal{T}_{x_i}\mathcal{M}$ is given by the $\hd$ largest left eigenvectors of the covariance matrix $\mathbf{R}_i$ from \eqref{cov_mat}, where $\hd$ is an estimate of $\dim(\M)$ or $\dim(\M)$ \textcolor{black}{itself}, if known.
\end{definition}
\textcolor{black}{When the true manifold dimension $d$ is not known, it is possible to estimate it directly from the sampled points. In the ideal case of neighboring points in $\mathcal{N}_i^P$ being located
exactly on $\mathcal{T}_{x_i}\mathcal{M}$, it holds that $\textrm{rank}(\mathbf{X}_i) = \textrm{rank}(\mathbf{B}_i) = d$, therefore  only $d$ singular values are non-vanishing. In this ideal case, $d$ can be obviously estimated as the number of singular values different from zero. However, there may usually be more than $d$ non-vanishing singular values due to
the curvature effect. In this case, it is possible to estimate the dimension $d$ as the number of singular values accounting for a certain (high) percentage of the variability of the displacements in $\mathbf{B}_i$. In practice,  denoting  the singular values of $\mathbf{B}_i$ with $\beta_{i,1} \geq \beta_{i,2} \geq \dots \geq \beta_{i,|\mathcal{N}_i^P|}$, a threshold
$0<\gamma\leq 1$ (possibly close to 1) is chosen and a local dimension $\hd_i$ is estimated per each point $x_i$ as the smallest number of singular values for which $\sum_{j = 1}^{\hd_i}\beta_{i,j}/\sum_{j = 1}^{|\mathcal{N}_i^P|}\beta_{i,j} \geq \gamma$.
For example, setting $\gamma = 0.8$ means that $\hd_i$ singular values account for at least $80\%$
variability of the displacements. At this point, the estimate $\hd$ of the dimension $d$ of the manifold is obtained as the (integer) mean or the median of the local estimated dimensions $\{\hd_i\}_{i=1}^n $\cite{singer2012vdm}}. Definition \ref{def:approxtangent} is equivalent to say that $\mathbf{O}_i$ is built with the first $\hd$ columns of $\mathbf{M}_i$ from \eqref{cov_mat}. Moreover, as usual, $\mathbf{O}_i$ can be equivalently (and efficiently) computed as the first $\hd$ left singular vectors of $\B_i$, without explicitly computing the covariance matrix $\R_i$. The local PCA procedure is summarized in Algorithm 1 in Section A of the Supplemental Material. Now, a discrete approximation of the parallel transport operator \cite{lee2006Riemannian}, that is a linear transformation from $\mathcal{T}_{x_i}\mathcal{M}$ to $\mathcal{T}_{x_j}\mathcal{M}$, is needed. In the discrete domain, this translates to associating a matrix to each edge of the above graph. For $\epsilon_n$  small enough, $\mathcal{T}_{x_i}\mathcal{M}$ and $\mathcal{T}_{x_j}\mathcal{M}$ are close, meaning that the column spaces of $\Oi$ and $\Oj$ are similiar. If the column spaces coincide, then the matrices $\Oi$ and $\Oj$ are \textcolor{black}{related by an orthogonal transformation $\widetilde{\mathbf{O}}_{i,j} = \Oi^T\Oj$.} However, if $\M$ is curved, the column spaces of $\Oi$ and $\Oj$ will not coincide. \textcolor{black}{For this reason, the transport operator approximation $\Oij$ is defined as the closest orthogonal matrix to $\widetilde{\mathbf{O}}_{i,j}$ \cite{singer2012vdm} \textcolor{black}{(Eq. 2.4, page 6)}, i.e.:}
\textcolor{black}{\begin{equation}\label{disc_trans_prob}
    \Oij = \argmin_{\mathbf{O}:\mathbf{O}^T\mathbf{O} = \mathbf{I}} \| \mathbf{O} - \widetilde{\mathbf{O}}_{i,j} \|_{HS},
\end{equation}}
\textcolor{black}{where $\| \cdot \|_{HS}$ is the Hilbert-Schmidt norm.} The solution of problem \eqref{disc_trans_prob} is given by $\Oij = \mathbf{M}_{i,j}\mathbf{V}^T_{i,j} \in \mathbb{R}^{\hd\times\hd}$, where $\mathbf{M}_{i,j}$ and $\mathbf{V}_{i,j}$ are the SVD of $\widetilde{\mathbf{O}}_{i,j} = \mathbf{M}_{i,j}\boldsymbol{\Sigma}_{i,j}\mathbf{V}^T_{i,j}$ (and the restriction maps of the approximating sheaf); a pictorial view of this discrete approximating transport is presented in Fig. \ref{fig:transport}. We now build a block matrix $\S \in \mathbb{R}^{n\hd\times n\hd}$ and a diagonal block matrix $\D \in \mathbb{R}^{n\hd\times n\hd}$ with $\hd \times \hd $ blocks defined as
\begin{align}
\label{DS}
&\S_{i,j} = w_{i,j}\widetilde{\mathbf{D}}_i^{-1}\Oij\widetilde{\mathbf{D}}_j^{-1}, \quad\D_{i,i} = \textrm{ndeg}(i) \mathbf{I}_{\hd} ,
\end{align}
where  $\widetilde{\mathbf{D}}_i = \textrm{deg}(i)\mathbf{I}_{\hd}$, $\textrm{deg}(i) = \sum_{j}w_{i,j}$ is the degree of node $i$, and $\textrm{ndeg}(i) = \sum_{j}w_{i,j}/(\textrm{deg}(i)\textrm{deg}(j))$ is the normalized degree of node $i$. Finally, we define the (normalized) Sheaf Laplacian as the following matrix
\begin{equation}
\label{sheaf_laplacian}
    \Delta_n = \epsilon_n^{-1}\big(\D^{-1}\S - \mathbf{I}\big) \in \mathbb{R}^{n\hd\times n\hd},
\end{equation}
which is the approximated Connection Laplacian of the underlying manifold \black{$\ccalM$} \cite{singer2012vdm} \textcolor{black}{(page 13)}. The procedure to build the Sheaf Laplacian is summarized in Algorithm 2 in Section A of the Supplemental Material. A sheaf $\tm_n$ with this (orthogonal) structure represents a discretized version of  $\tm$. Further details in \cite{singer2012vdm}.

At this point, we introduce a linear sampling operator $\samp_n^{\mathcal{X}}:\ltm \rightarrow \ltmn$ to discretize a  tangent bundle signal $\F$ as a sheaf signal $\f_n \in \mathbb{R}^{n\hd}$ such that (refer to Appendix \ref{ap:proofth1} for the rigorous definition  of $\ltmn$):
\begin{align}
\label{sampler}
    &\f_n = \samp_n^{\mathcal{X}}\F, \\
    &\f_{n}(x_i):=[\f_n]_{((i-1)\hd+1):(i+1)\hd} = \Oi^T  d\i\F(x_i) \in \mathbb{R}^\hd,
\end{align}
where $((i-1)\hd+1):(i+1)\hd$ indicates all the components of $\f_{n}$ from the $((i-1)\hd+1)$-th to the $(i+1)\hd$-th component. In words, the sampling operator $\samp_n^{\mathcal{X}}$ in \eqref{sampler} takes  the embedded tangent signal $d\i\F$ as input, evaluates it on each point $x_i$ in the sampling set $\mathcal{X}$, projects the evaluated signals $d\i_{x_i}\left(\F(x_i)\right) \in \mathbb{R}^p$ over the $d$-dimensional subspaces spanned by the $\Oi$s from Definition \ref{def:approxtangent} and, finally, sequentially collects the $n$ projections $\Oi^Td\i\F(x_i) \in \mathbb{R}^\hd$ in the vector $\f_n \in \mathbb{R}^{n\hd}$, representing the discretized tangent bundle signal. We are now in the condition of plugging the discretized operator from \eqref{sheaf_laplacian} and signal from \eqref{sampler} in the definition of tangent bundle filter from \eqref{param_conv}, obtaining:
\begin{equation}
    \label{discr_param_conv}
    \g_n = \int_0^{\infty}\th(t)e^{t\Delta_n}\f_n\textrm{d}t = \h(\Delta_n)\f_n \in \mathbb{R}^{n\hd}.
\end{equation}
Following the same considerations of Section \ref{sec:tnn}, we can define a  discretized space tangent bundle neural network (D-TNN) as the stack of $L$ layers of the form
\begin{equation}
    \label{dt_tnn_layer}
    \f_{n,l+1}^u = \sigma\Bigg(\sum_{q=1}^{F_l}\h(\Delta_n)_l^{u,q}\f_{n,l}^q\Bigg), \; u = 1,...,F_{l+1},
\end{equation}
where (with a slight abuse of notation) $\sigma$ has the same point-wise law of $\widetilde{\sigma}$ in Definition \ref{def:pointwise}.  As in the continuous case, we describe the $u$th output of a D-TNN as a mapping $\bPsi_u\big(\mathcal{H}, \Delta_n, \{\x_n^q\}_{q=1}^{F_0}\big)$ to emphasize that it is parameterized by filters $\mathcal{H}$ and the Sheaf Laplacian $\Delta_n$. The D-TNN architecture comes with desirable theoretical properties. As the number of sampling points goes to infinity, the Sheaf Laplacian $\Delta_n$ converges to the Connection Laplacian $\Delta$ \cite{singer2012vdm} and the sheaf signal $\x_n$ consequently converges to the tangent bundle signal $\F$. Combining these results, we prove in the next theorem that the output of a D-TNN converges to the output of the corresponding underlying TNN as the sample size increases, validating the approximation fitness of a D-TNN. To the best of our knowledge, this is the first result to \textit{formally} connect Sheaf Neural Networks to tangent bundles of Riemann manifolds. Let us denote the injectivity radius and the condition number \cite{singer2017spectral,lee2006Riemannian} of the manifold $\ccalM$ with $\kappa$ and $\tau$, respectively.

\begin{theorem} \label{theorem-main}
    Let $\mathcal{X}=\{x_1,\dots,x_n\}\subset \mathbb{R}^p$ be a set of $n$ i.i.d. sampled points from measure $\mu$ over $\M \subset \mathbb{R}^p$ and $\F$ a tangent bundle signal. Let $\tm_n$ be the cellular sheaf built from $\mathcal{X}$ as explained above \textcolor{black}{with $\hd=d$ and $0<\epsilon_n \leq \min\{\kappa, \tau^{-1}\}$}. Let  $\bPsi_u\big(\mathcal{H}, \cdot, \cdot \big)$ be the $u$th output of a neural network of $L$ layers parameterized by the operator $\Delta$ of  $\tm$  or by the discrete operator $\Delta_n$ of  $\tm_n$. If:
\begin{description}
    \item[\, A1] the frequency response of filters in $\mathcal{H}$ are non-amplifying Lipschitz continuous;
    \item[\, A2]  \textcolor{black}{Each filter $\widetilde{h}(\cdot) \in \mathcal{H}$ is a lowpass filter};
    \item[\, A3] $\widetilde{\sigma}$ from Definition \ref{def:pointwise} is point-wise normalized Lipschitz continuous,
\end{description}
then \textcolor{black}{there exists a sequence of scales $\epsilon_n \rightarrow 0$ as $n \rightarrow \infty$ s.t.}
\begin{equation}
\label{convergence_main}
\lim_{n \rightarrow \infty} ||\bPsi_u\big(\mathcal{H}, \Delta_n, \samp_n^{\mathcal{X}}\F\big) - \samp_n^{\mathcal{X}}\bPsi_u\big(\mathcal{H}, \Delta,\F\big)||_{\tm_n} = 0,
\end{equation}
with the limit in probability, for each $u = 1, 2, \dots, F_L$.
\end{theorem}
\begin{proof}
    See Appendix \ref{ap:proofth1}.
\end{proof}
\begin{remark}
    \textcolor{black}{Denoting the Sheaf Laplacian with $\Delta_n$ is an abuse of notation, because Theorem 1 is a condition both on $\epsilon_n\to 0$ and $n\to \infty$. For this reason, we should employ a notation such as $\Delta_{n,\epsilon_n}$; however, we will keep $\Delta_n$ in the following for the sake of exposition and consistency.}
\end{remark}
\begin{table*}[t]
    \centering
    \scalebox{.75}{\textcolor{black}{\begin{tabular}{|c|c|c|}
    \hline
        &Tangent Bundle $\tm$ & Cellular Sheaf $\tm_n$ \\
        \hline\xrowht[()]{4pt}
        Signal & $\F$ &  $\f_n$ \\
        \hline\xrowht[()]{4pt}
        Laplacian &$\Delta$ & $\Delta_n$ \\
        \hline\xrowht[()]{7pt}
    Inner product & $\langle \F, \G \rangle_{\tm} = \int_{\M} \langle \F(x), \mathbf{G}(x) \rangle_{\tmx} \textrm{d}\mu(x)$ & $\langle\f_n,\g_n\rangle_{\tm_n}=\frac{1}{n}\sum_{i = 1}^n  \f_{n}(x_i) \dotp \mathbf{u}_{n}(x_i)$ \\
        \hline\xrowht[()]{7pt}
    Filter & $\G = \int_0^{\infty}\th(t)e^{t\Delta}\F\textrm{d}t = \h(\Delta)\F$ &  $\mathbf{g}_n=  \sum_{k=0}^{K-1} h_k e^{k\Delta_n}\f_n= \h(\Delta_n)\f$\\
        \hline\xrowht[()]{11pt}
    Neural Network &  $\F_{l+1}^u = \sigma\Big(\sum_{q=1}^{F_l}\h(\Delta)_l^{u,q}\F_l^q\Big)$ & $\f_{n,l+1}^u = \sigma\Big(\sum_{q=1}^{F_l}\h(\Delta_n)_l^{u,q}\f_{n,l}^q\Big)$
    \tabularnewline
    \hline
    \end{tabular}}}
    \caption{\textcolor{black}{Notation and summary of the continuous (on tangent bundles) framework and its discretization (on cellular sheaves)}}
    \label{wrap_Table}
\end{table*}
\textcolor{black}{
Theorem 1 requires the filters to be lowpass. This can be challenging in a learning context because the filters are learned end-to-end and they may or may not satisfy this hypothesis. Thus, the practical implication of Theorem 1 is that it is possible to train TNNs on sampled manifolds although we do not offer an explicit method to guarantee that this is indeed attained. A first important point to make is that this condition is not spurious, as it is a minimal condition imposed in the proof of Theorem 1 to guarantee convergence. A second important point is that filters can be forced to be lowpass by constraining the filters coefficients during training, if needed. Here we do not advocate the use of these constraints.}

\vspace{-.2cm}
\subsection{Discretization in the Time Domain}\label{sec:time_dis}
The discretization in space introduced in the previous section is still not enough for implementing TNNs in practice. Indeed, learning the continuous time function $\tilde{h}(t)$ is in general infeasible. For this reason, we discretize $\tilde{h}(t)$ in the continuous time domain by fixing a sampling interval $T_s>0$. In this way, we can replace the filter response function with a series of coefficients $h_k = \tilde{h}(k T_s)$, $k =0 ,1, 2\dots$. Fixing $T_s=1$ and taking $K$ samples over the time horizon, the discrete-time version of the convolution in \eqref{convolution} is given by
\begin{equation}
\label{eqn:manifold_convolution_discrete}
    \bbh(\Delta_n) \F(x)= \sum_{k=0}^{\infty} h_k e^{k\Delta_n}\F(x), 
\end{equation}
which can be seen as a finite impulse response (FIR) filter with shift operator $e^{\Delta_n}$. We are now in the condition of injecting the space discretization from Section \ref{sec:disc} in the finite-time architecture in \eqref{eqn:manifold_convolution_discrete}, thus finally obtaining an implementable tangent bundle filter that exploits the approximating cellular sheaf $\tm_n$ as
\begin{equation}
\label{eqn:discrete_manifold_convolution_discrete}
  \mathbf{g}_n=  \bbh(\Delta_n) \f_n = \sum_{k=0}^{K-1} h_k e^{k\Delta_n}\f_n.
\end{equation}
The discretized manifold filter of order $K$ can be seen as a generalization of graph convolution to the orthogonal cellular sheaf domain. Thus, we refer $e^{\Delta_n}$ as a sheaf shift operator. At this point, by replacing the filter $\bbh_l^{pq}(\Delta_n)$ in \eqref{dt_tnn_layer} with \eqref{eqn:discrete_manifold_convolution_discrete}, we obtain the following architecture:
\begin{equation}
    \label{sheaf_nn_layer}
    \f_{n,l+1}^u = \sigma\Bigg(\sum_{q=1}^{F_l}\sum_{k = 0}^{K-1}h_{k,l}^{u,q}\big(e^{\Delta_n}\big)^k\f_{n,l}^q\Bigg), \; u = 1,...,F_{l+1},
\end{equation}
that we refer to as discretized space-time tangent bundle neural network (DD-TNN). DD-TNNs are a novel principled variant of the recently proposed Sheaf Neural Networks \cite{hansen2020sheafnn,bodnar2022sheafdiff,barbero2022sheafnnconn}, with $e^{\Delta_n}$  as (sheaf) shift operator and order $K$ diffusion. To better enhance this similarity, we rewrite the layer in \eqref{sheaf_nn_layer} in  matrix form by introducing the matrices $\mathbf{X}_{n,l}=\{\f_{n,l}^u\}_{u=1}^{F_{l}}\in \mathbb{R}^{n\hd \times F_{l}}$, and $\mathbf{H}_{l,k} = \{h_{k,l}^{u,q}\}_{q=1,u=1}^{F_l,F_{l+1}}\in \mathbb{R}^{F_l \times F_{l+1}}$, as
\begin{equation}\label{tnn_matrix}
    \mathbf{X}_{n,l+1} = \sigma\Bigg(\sum_{k = 0}^{K-1}\big(e^{\Delta_n}\big)^k\mathbf{X}_{n,l}\mathbf{H}_{l,k}\Bigg) \; \in \mathbb{R}^{n\hd \times F_{l+1}},
\end{equation}
where the filter weights $\{\mathbf{H}_{l,k}\}_{l,k}$ are learnable parameters. Finally, we have completed the process of building TNNs from (orthogonal) cellular sheaves and back.  The proposed methodology also shows that manifolds and their tangent bundles can be seen as the limits of graphs and  (orthogonal) cellular sheaves on top of them. A summary of the proposed continuous framework on tangent bundles and its discretization on orthogonal cellular sheaves is presented in Table \ref{wrap_Table}. \textcolor{black}{Please notice that, when $T_s=1$ and $K=1$ in \eqref{tnn_matrix}, the standard Sheaf Neural Network from \cite{hansen2020sheafnn} (up to an additional channel mixing matrix) with the exponential of the sheaf Laplacian as shift operator is recovered.}
\vspace{-.1cm}

\section{Consistency of Tangent Bundle Convolutions}\label{sec:consistency} 

\textcolor{black}{The tangent bundle convolution in Definition \ref{def:tangent-bundle-filter} provides a definition of a convolution that is compatible with convolutions on manifold scalar fields, convolutions on graphs, and (standard) convolutions for signals in time.}

\textcolor{black}{The manifold convolution from \cite{wang2022convolution} is recovered when the bundle is a scalar bundle, i.e. when scalar functions $f:\mathcal{M}\rightarrow \mathbb{R}$ over the manifold are considered. In this case, the Connection Laplacian $\Delta$ reduces to the usual Laplace-Beltrami operator \cite{lee00manif}, here denoted with $\Delta_{\mathcal{M}}: \mathcal{L}_2(\mathcal{M})\rightarrow\mathcal{L}_2(\mathcal{M})$, and the resulting convolution, given a filter $\widetilde{h}:\mathbb{R}^+\rightarrow \mathbb{R}$, is
\begin{equation}
    g(x) = \big(\widetilde{h} \star_{\mathcal{M}} f)(x) = \int_0^{\infty}\widetilde{h}(t)e^{-t\Delta_\mathcal{M}}f(x)\textrm{d}t.
\end{equation}
This expression is both the manifold convolution from \cite{wang2022convolution} and a particular case of \eqref{param_conv}. The negative sign comes from the convention to define the standard Laplacian as a positive semidefinite operator whereas the Connection Laplacian is defined to be negative semidefinite.}

\textcolor{black}{Given a set of $n$ points $\mathcal{X} \subset \mathbb{R}^p$ sampled from the manifold, we further recover a form of graph convolution \cite{shuman2013emerging, gama2018convolutional}. In particular, if the manifold $\mathcal{M}$ is discretized as a geometric graph $\mathcal{M}_n$ whose nodes are the sampled points, the Laplace-Beltrami operator $\Delta_{\mathcal{M}}$ is discretized as a graph Laplacian $\Delta_{\mathcal{M},n} \in \mathbb{R}^{n\times n}$ whose entries are the weights $w_{i,j}$ of equation \eqref{graph_weights}.  If we further discretize $f:\mathcal{M}\rightarrow \mathbb{R}$ as a graph signal $\mathbf{f}_n:\mathcal{M}_n\rightarrow \mathbb{R}$, the resulting convolution is
\begin{equation}\label{graph_conv}
   \mathbf{g}_n=\big(\widetilde{h} \star_{\mathcal{M}_n} \mathbf{f}_n) = \int_0^{\infty}\widetilde{h}(t)e^{-t\Delta_{\mathcal{M},n}}\mathbf{f}_n\textrm{d}t.
\end{equation}
This is a particular case of \eqref{discr_param_conv} and can be interpreted as an exponential form of a graph convolution. Further discretizing the filter across the index $t$ as we do in Section IV.C. yields the graph convolution
\begin{equation} \label{graph_conv_dis}
  \mathbf{g}_n=  \sum_{k=0}^{K-1} h_k (e^{-\Delta_{\mathcal{M},n}})^k\mathbf{f}_n.
\end{equation}
This is a FIR graph filter with $e^{-\Delta_{\mathcal{M},n}}$ used as a shift operator. The expression can be made more familiar if we approximate the exponential by $e^{-\Delta_{\mathcal{M},n}}\approx \mathbf{I}_n -\Delta_{\mathcal{M},n}$.}

\textcolor{black}{The standard time convolution is recovered when the manifold is the real line $\mathbb{R}$, the functions $f:\mathbb{R}\rightarrow \mathbb{R}$ are scalar functions, and the operator employed in the heat equation in \eqref{diff_eq} is replaced by the derivative operator $\partial/\partial x$. In particular, due to the fact that the exponential of the derivative operator is a time shift operator, we can write $e^{-t\partial/\partial x}f(x) = f(x-t)$. In this case, the resulting convolution is
\begin{align}\label{eqn_continuous_time_convolutions}
    g(x) = \big(\widetilde{h} \star_{\mathbb{R}} f)(x) &= \int_0^{\infty}\widetilde{h}(t)e^{-t\partial/\partial x}f(x)\textrm{d}t \nonumber \\
    &= \int_{0}^\infty \widetilde{h}(t) f(x-t) \,\text{d}t,
\end{align}
This is the (standard) time convolution and also a particular case of \eqref{param_conv}.} \textcolor{black}{An additional amenable theoretical feature of our tangent bundle convolution is its consistency with the framework of Algebraic Signal Processing (ASP) \cite{parada2020algebraic,puschel2008algebraic}. An ASP model is made of four components: (i) A vector space $\mathbb{V}$ where the signals of interest live. (ii) The space $\text{End}(\mathbb{V})$ of endomorphisms of $\mathbb{V}$ containing the linear maps that can be applied to the signals in $\mathbb{V}$. (iii) An Algebra $\mathbb{A}$ that defines abstract convolutional filters. (iv) A homomorphism $\rho$ that maps filters in $\mathbb{A}$ to endomorphisms that can be applied to signals. 
In our case, the vector space $\mathbb{V}$ is made of tangent bundle signals, and the algebra $\mathbb{A}$ is the algebra $(\mathcal{L}_1(\mathbb{R}_{+}),\star_\mathbb{R} )$ of absolute integrable functions in $\mathbb{R}_{+}$  with the standard convolution $\star_\mathbb{R}$ as the product. The homomorphism $\rho$ maps the filter $\widetilde{h}(t)$ to the tangent bundle filter $\rho(\widetilde{h})$ whose action on a signal $\mathbf{F}$ is
\begin{align}\label{eqn_asp_response_1}
   \rho(\widetilde{h}) \circ \mathbf{F} (s)
      = \int_0^\infty 
           \widetilde{h}(t) e^{t \Delta} \mathbf{F} (s) \, dt . 
\end{align}
This is clearly the definition we obtain in \eqref{param_conv} by combining \eqref{exp_sol} and \eqref{convolution}. It is trivial to verify that $\rho(\widetilde{h})$ is a homomorphism.}

\textcolor{black}{
An alternative definition of tangent bundle convolution would be obtained if we replace the algebra $\mathbb{A}$ with the algebra of polynomials. Thus, filters would be polynomials $\widetilde{h}(t) = \sum_{k=0}^{K} h_k t^k$ and the tangent bundle filters would be 
\begin{align}\label{eqn_asp_response_2}
   \rho(\widetilde{h}) \circ \mathbf{F} (s)
      = \sum_{k=0}^{K} 
           h_k \Delta^k  \mathbf{F} (s) .
\end{align}
In this latter case, discretizing the manifold would give rise to graph filters defined as polynomials of the graph Laplacian. In this paper we prefer to work with \eqref{eqn_asp_response_1} rather than \eqref{eqn_asp_response_2} because it leads to the connection with convolutions in continuous time stated in \eqref{eqn_continuous_time_convolutions}. This connection can't be made if we adopt  \eqref{eqn_asp_response_2} as a definition of tangent bundle filter. It is important to remark that if we adopt \eqref{eqn_asp_response_2} as a definition a similar convergence theorem holds. We just need to change the definition of the filter's frequency response to the polynomial $\sum_{k=0}^{K} h_k\lam^k$ and proceed to adapt assumptions and derivations.}
\vspace{-.2cm}
\section{Numerical Results}\label{sec:num_res}
In this section, we assess the performance of Tangent Bundle Neural Networks on \textcolor{black}{four} tasks: denoising of a tangent vector field on the torus (synthetic data),  reconstruction from partial observations of the Earth wind field (real data), forecasting of the Earth wind field (real data), obtained via a recurrent version of the proposed architecture,  and \textcolor{black}{binary manifold classification (synthetic data)}. In this work, we are interested in showing the advantage of including information about the tangent bundle structure for processing tangent bundle signals. For this reason, in the following experiments we always use the vanilla DD-TNN architecture in \eqref{tnn_matrix} without any additional modules \textcolor{black}{(e.g.~readout MLP layers)}, and we compare our architectures against vanilla Manifold Neural Networks (MNNs) from \cite{wang2022convolution}, convolutional architectures built in a similar way to ours but taking into account only the manifold structure. MNNs are implemented as GNNs with the exponential of the normalized cloud Laplacian \cite{belkin2001laplacian, wang2022convolution}. \textcolor{black}{Moreover, we also compare DD-TNNs against Multi-Layer Perceptrons (MLPs) \cite{haykin1994neural} in the denoising and reconstruction tasks, against Recurrent Neural Networks (RNNs) \cite{rumelhart1985learning} in the forecasting task, and against 3D-CNN in the classification task. Therefore, from a discrete point of view, we present a comparison between a specific (novel and principled) Sheaf Neural Networks class (DD-TNNs, which introduce a relational inductive bias \cite{battaglia2018relational} given by the tangent bundle/sheaf structure), a specific Graph Neural Networks class (MNNs, which introduce a relational inductive bias given by the manifold/graph structure), and Multi-Layer Perceptrons/Recurrent Neural Networks (MLPs/RNNs, which introduce no relational inductive biases).}  It is clear that the employed classes of architectures could be enriched with many additional components (biases, layer normalization, dropout, gating, just to name a few), and it is also clear that a huge number of other architectures could be tailored to the proposed tasks, but testing them is beyond the scope of this paper.\footnote{Our implementation of TNNs \& datasets available at \url{https://github.com/clabat9/Tangent-Bundle-Neural-Networks}}
\begin{table*}[t!]
\centering
\scalebox{.75}{\begin{tabular}{|l|cccc|}
\hline
& & $\tau =  10^{-2}$ & $\tau = 10^{-1}$ & $\tau = 3 \cdot 10^{-1}$ \\
\hline
\multicolumn{1}{|l|}{\multirow{3}{*}{$\textrm{E}\{n\}=100$}} & \multicolumn{1}{c|}{DD-TNN} & \multicolumn{1}{c|}{$\mathbf{2.02\cdot 10^{-4}} \pm 1.88\cdot 10^{-5}$} & \multicolumn{1}{c|}{$\mathbf{1.78\cdot 10^{-2}} \pm 1.96 \cdot 10^{-3}$} & \multicolumn{1}{c|}{$\mathbf{1.35\cdot 10^{-1}} \pm 1.42 \cdot 10^{-2}$} \\
\multicolumn{1}{|l|}{} & \multicolumn{1}{c|}{MNN} & \multicolumn{1}{c|}{$7.33\cdot 10^{-4} \pm 4.61 \cdot 10^{-4}$} & \multicolumn{1}{c|}{$2.45\cdot 10^{-2} \pm 4.26 \cdot 10^{-3}$} & \multicolumn{1}{c|}{$2.19\cdot 10^{-1} \pm 3.56 \cdot 10^{-2}$} \\
\multicolumn{1}{|l|}{} & \multicolumn{1}{c|}{\textcolor{black}{MLP}} & \multicolumn{1}{c|}{$2.34\cdot 10^{-4} \pm 2.88 \cdot 10^{-5}$} & \multicolumn{1}{c|}{$1.83\cdot 10^{-2} \pm 2.48 \cdot 10^{-3}$} & \multicolumn{1}{c|}{$1.52\cdot 10^{-1} \pm 2.15 \cdot 10^{-2}$} \\
\hline
\multicolumn{1}{|l|}{\multirow{3}{*}{$\textrm{E}\{n\}=200$}} & \multicolumn{1}{c|}{DD-TNN} & \multicolumn{1}{c|}{$\mathbf{2.06\cdot 10^{-4}} \pm 1.46\cdot 10^{-5}$} & \multicolumn{1}{c|}{$\mathbf{1.82\cdot 10^{-2}} \pm 1.18 \cdot 10^{-3}$} & \multicolumn{1}{c|}{$\mathbf{1.36\cdot 10^{-1}} \pm 1.05 \cdot 10^{-2}$} \\
\multicolumn{1}{|l|}{} & \multicolumn{1}{c|}{MNN} & \multicolumn{1}{c|}{$7.78\cdot 10^{-4} \pm 5.76 \cdot 10^{-4}$} & \multicolumn{1}{c|}{$2.50\cdot 10^{-2} \pm 3.90 \cdot 10^{-3}$} & \multicolumn{1}{c|}{$2.11\cdot 10^{-1} \pm 3.30 \cdot 10^{-2}$} \\
\multicolumn{1}{|l|}{} & \multicolumn{1}{c|}{\textcolor{black}{MLP}} & \multicolumn{1}{c|}{$2.28\cdot 10^{-4} \pm 3.52 \cdot 10^{-5}$} & \multicolumn{1}{c|}{$1.88\cdot 10^{-2} \pm 2.88 \cdot 10^{-3}$} & \multicolumn{1}{c|}{$1.55\cdot 10^{-1} \pm 2.06 \cdot 10^{-2}$} \\
\hline
\multicolumn{1}{|l|}{\multirow{3}{*}{$\textrm{E}\{n\}=300$}} & \multicolumn{1}{c|}{DD-TNN} & \multicolumn{1}{c|}{$\mathbf{2.05\cdot 10^{-4}} \pm 1.07\cdot 10^{-5}$} & \multicolumn{1}{c|}{$\mathbf{1.80\cdot 10^{-2}} \pm 1.01 \cdot 10^{-3}$} & \multicolumn{1}{c|}{$\mathbf{1.31\cdot 10^{-1}} \pm 7.91 \cdot 10^{-3}$} \\
\multicolumn{1}{|l|}{} & \multicolumn{1}{c|}{MNN} & \multicolumn{1}{c|}{$6.64\cdot 10^{-4} \pm 4.13 \cdot 10^{-4}$} & \multicolumn{1}{c|}{$2.43\cdot 10^{-2} \pm 4.01 \cdot 10^{-3}$} & \multicolumn{1}{c|}{$2.05\cdot 10^{-1} \pm 3.06 \cdot 10^{-2}$} \\
\multicolumn{1}{|l|}{} & \multicolumn{1}{c|}{\textcolor{black}{MLP}} & \multicolumn{1}{c|}{$2.36\cdot 10^{-4} \pm 2.87\cdot 10^{-5}$} & \multicolumn{1}{c|}{$1.85\cdot 10^{-2} \pm 2.25 \cdot 10^{-3}$} & \multicolumn{1}{c|}{$1.51\cdot 10^{-1} \pm 1.87 \cdot 10^{-2}$} \\
\hline
\multicolumn{1}{|l|}{\multirow{3}{*}{$\textrm{E}\{n\}=400$}} & \multicolumn{1}{c|}{DD-TNN} & \multicolumn{1}{c|}{$\mathbf{2.00\cdot 10^{-4}} \pm 9.60\cdot 10^{-6}$} & \multicolumn{1}{c|}{$\mathbf{1.80\cdot 10^{-2}} \pm 8.99 \cdot 10^{-4}$} & \multicolumn{1}{c|}{$\mathbf{1.35\cdot 10^{-1}} \pm 8.03 \cdot 10^{-3}$} \\
\multicolumn{1}{|l|}{} & \multicolumn{1}{c|}{MNN} & \multicolumn{1}{c|}{$6.84\cdot 10^{-4} \pm 6.28 \cdot 10^{-4}$} & \multicolumn{1}{c|}{$3.45\cdot 10^{-2} \pm 5.88 \cdot 10^{-2}$} & \multicolumn{1}{c|}{$2.55\cdot 10^{-1} \pm 9.50 \cdot 10^{-2}$} \\
\multicolumn{1}{|l|}{} & \multicolumn{1}{c|}{\textcolor{black}{MLP}} & \multicolumn{1}{c|}{$2.26\cdot 10^{-4} \pm 3.27\cdot 10^{-5}$} & \multicolumn{1}{c|}{$1.86\cdot 10^{-2} \pm 2.28 \cdot 10^{-3}$} & \multicolumn{1}{c|}{$1.58\cdot 10^{-1} \pm 1.90 \cdot 10^{-2}$} \\
\hline
\end{tabular}}
\caption{MSE on the torus denoising task}
\label{table:results_torus}
\end{table*}
\vspace{-.3cm}
\subsection{Torus Denoising}\label{torus_den}
\vspace{-.1cm}
We design a denoising task on a 2-dimensional torus ($\ccalM=\mathcal{T}_2$) and its tangent bundle. \textcolor{black}{It is well known that the $2$-torus, the $2$-sphere, the real projective plane, together with their connected sums completely classify closed $2$-dimensional manifolds, thus it is a good manifold to test our architecture.} A \textcolor{black}{parameterization of the} 2-dimensional torus is obtained by revolving a circle in three-dimensional space about an axis that is coplanar with the circle:
$[x,y,z] = [(b+a\cos{\theta})\cos{\phi},(b+a\cos{\theta})\sin{\phi},r\sin{\theta}]$,
where $\phi,\theta \in [0, 2\pi)$, $a$ is the radius of the tube, and  $b$ is the distance from the center of the tube to the center of the torus; $b/a$ is called the aspect ratio. In this experiment, we work on a ring torus, thus a torus with aspect ratio greater than one (in particular, we choose $b= 0.3$, $a = 0.1$). We uniformly sample the torus on $n$ points $\mathcal{X}=\{x_1,\dots,x_n\}$, and we compute the corresponding cellular sheaf $\tm_n$, Sheaf Laplacian $\Delta_n$ and signal sampler $\samp_n^{\mathcal{X}}$ as explained in Section \ref{sec:space_disc}, with $\epsilon_{\textrm{PCA}}=0.8$ and $\epsilon_n = 0.5$. We consider the tangent vector field over the torus given by $d\i\F(x,y,z)=(-\sin{\theta},\cos{\theta},0) \in 
\mathbb{R}^3$.
At this point, we add AWGN with variance $\tau^2$ to $ d\i\F$ obtaining a noisy field $\widetilde{ d\i\F}$, then we use $\samp_n^{\mathcal{X}}$ to sample it, obtaining $\widetilde{\mathbf{f}}_n \in \mathbb{R}^{2n}$. We test the performance of the TNN architecture by evaluating its ability to denoising $\widetilde{\mathbf{f}}_n$. We exploit a 3 layers architecture with $8$ and $4$ hidden features, and $1$ output feature (the denoised signal),  using $K=2$ in each layer, with $\textrm{Tanh()}$ non-linearities in the hidden layers and a linear activation on the output layer; \textcolor{black}{the architecture hyperparameters have been chosen with hyperparameters sweeps}. We train the architecture to minimize the square error $\|\widetilde{\mathbf{f}}_n - \mathbf{f}_{n}^o\|^2$ between the noisy signal $\widetilde{\mathbf{f}}_n$ and the output of the network $\mathbf{f}_{n}^o$ via the ADAM optimizer \cite{kingma2014adam} and a patience of 5 epochs, with hyperparameters set to obtain the best results. We compare our architecture with a 3 layers MNN (implemented via a GNN as explained in \cite{wang2022convolution}) with same hyperparameters; to make the comparison fair, $\widetilde{ d\i\F}$ evaluated on $\mathcal{X}$ is given as input to the MNN, organized in a  matrix $\widetilde{\F}_n \in \mathbb{R}^{n \times 3}$. We train the MNN to minimize the square error $\|\widetilde{\F}_n - \mathbf{F}_{n}^o\|_F^2$, where $\|\|_F$ is the Frobenius Norm and $\mathbf{F}_{n}^o$ is the network output. It is trivial to see that the "two" MSEs used for TNN and MNN are completely equivalent due to the orthogonality of the projection matrices $\Oi$. In Table \ref{table:results_torus}, we evaluate TNNs and MNNs for four different expected sample sizes ($\textrm{E}\{n\} = 100$, $\textrm{E}\{n\} = 200$, $\textrm{E}\{n\} = 300$, and $\textrm{E}\{n\}=400$), for three different noise standard deviation ($\tau = 10^{-2}$,$\tau = \cdot10^{-1}$ and $\tau = 3 \cdot 10^{-1}$), showing the MSEs $\frac{1}{n}\|\mathbf{f}_n - \mathbf{f}_{n}^o\|^2$ and $\frac{1}{n}\|\F_n - \mathbf{F}_{n}^o\|_F^2$, where $\mathbf{f}_n$ is the sampling via $\samp_n^{\mathcal{X}}$ of the clean field and $\mathbf{F}_n$ is the matrix collecting the clean field evaluated on $\mathcal{X}$. 8 sampling realizations and 8 mask realizations per each of them are tested; to make the results consistent, divergent or badly trained runs are discarded if present, and then the results are averaged \textcolor{black}{(on average about 2 runs are discarded per each sampling realization)}. As the reader can notice from Table 1, TNNs always perform better than MNNs and MLPs, due to their "bundle-awareness", i.e. the sheaf structure. 
\begin{table*}[]
\centering
\scalebox{.75}{\begin{tabular}{|l|cccc|}
\hline
& & $\textrm{E}\{\widetilde{n}\} = 0.5n$ & $\textrm{E}\{\widetilde{n}\} = 0.3n$ & $\textrm{E}\{\widetilde{n}\} = 0.1n$ \\
\hline
\multicolumn{1}{|l|}{\multirow{3}{*}{$\textrm{E}\{n\}=100$}} & \multicolumn{1}{c|}{DD-TNN} & \multicolumn{1}{c|}{$\mathbf{1.93\cdot 10^{-2}} \pm 3.64\cdot 10^{-3}$} & \multicolumn{1}{c|}{$\mathbf{1.15\cdot 10^{-2}} \pm 2.75 \cdot 10^{-3}$} & \multicolumn{1}{c|}{$\mathbf{3.31\cdot 10^{-3}} \pm 1.62 \cdot 10^{-3}$} \\
\multicolumn{1}{|l|}{} & \multicolumn{1}{c|}{MNN} & \multicolumn{1}{c|}{$4.20\cdot 10^{-2} \pm 3.05 \cdot 10^{-2}$} & \multicolumn{1}{c|}{$3.12\cdot 10^{-2} \pm 1.86 \cdot 10^{-2}$} & \multicolumn{1}{c|}{$2.82\cdot 10^{-2} \pm 2.32 \cdot 10^{-2}$} \\
\multicolumn{1}{|l|}{} & \multicolumn{1}{c|}{\textcolor{black}{MLP}} & \multicolumn{1}{c|}{$2.00\cdot 10^{-2} \pm 3.99 \cdot 10^{-3}$} & \multicolumn{1}{c|}{$1.21\cdot 10^{-2} \pm 2.50 \cdot 10^{-3}$} & \multicolumn{1}{c|}{$3.61\cdot 10^{-3} \pm 1.70 \cdot 10^{-3}$} \\
\hline
\multicolumn{1}{|l|}{\multirow{3}{*}{$\textrm{E}\{n\}=200$}} & \multicolumn{1}{c|}{DD-TNN} & \multicolumn{1}{c|}{$\mathbf{1.99\cdot 10^{-2}} \pm 2.30\cdot 10^{-3}$} & \multicolumn{1}{c|}{$\mathbf{1.18\cdot 10^{-2}} \pm 1.62 \cdot 10^{-3}$} & \multicolumn{1}{c|}{$\mathbf{3.67\cdot 10^{-3}} \pm 1.23 \cdot 10^{-3}$} \\
\multicolumn{1}{|l|}{} & \multicolumn{1}{c|}{MNN} & \multicolumn{1}{c|}{$3.19\cdot 10^{-2} \pm 1.31 \cdot 10^{-2}$} & \multicolumn{1}{c|}{$2.74\cdot 10^{-2} \pm 1.55 \cdot 10^{-2}$} & \multicolumn{1}{c|}{$2.58\cdot 10^{-2} \pm 1.82 \cdot 10^{-2}$} \\
\multicolumn{1}{|l|}{} & \multicolumn{1}{c|}{\textcolor{black}{MLP}} & \multicolumn{1}{c|}{$2.03\cdot 10^{-2} \pm 2.28 \cdot 10^{-3}$} & \multicolumn{1}{c|}{$1.20\cdot 10^{-2} \pm 1.68 \cdot 10^{-3}$} & \multicolumn{1}{c|}{$3.69\cdot 10^{-3} \pm 1.17 \cdot 10^{-3}$} \\
\hline
\multicolumn{1}{|l|}{\multirow{3}{*}{$\textrm{E}\{n\}=300$}} & \multicolumn{1}{c|}{DD-TNN} & \multicolumn{1}{c|}{$\mathbf{1.88\cdot 10^{-2}} \pm 1.72\cdot 10^{-3}$} & \multicolumn{1}{c|}{$\mathbf{1.13\cdot 10^{-2}} \pm 1.54 \cdot 10^{-3}$} & \multicolumn{1}{c|}{$\mathbf{3.96\cdot 10^{-3}} \pm 1.00 \cdot 10^{-3}$} \\
\multicolumn{1}{|l|}{} & \multicolumn{1}{c|}{MNN} & \multicolumn{1}{c|}{$2.68\cdot 10^{-2} \pm 7.64 \cdot 10^{-3}$} & \multicolumn{1}{c|}{$2.41\cdot 10^{-2} \pm 1.05 \cdot 10^{-2}$} & \multicolumn{1}{c|}{$2.09\cdot 10^{-2} \pm 1.76 \cdot 10^{-2}$} \\
\multicolumn{1}{|l|}{} & \multicolumn{1}{c|}{\textcolor{black}{MLP}} & \multicolumn{1}{c|}{$1.95\cdot 10^{-2} \pm 1.74 \cdot 10^{-3}$} & \multicolumn{1}{c|}{$1.18\cdot 10^{-2} \pm 1.56 \cdot 10^{-3}$} & \multicolumn{1}{c|}{$4.00\cdot 10^{-3} \pm 8.85 \cdot 10^{-4}$} \\
\hline
\multicolumn{1}{|l|}{\multirow{3}{*}{$\textrm{E}\{n\}=400$}} & \multicolumn{1}{c|}{DD-TNN} & \multicolumn{1}{c|}{$\mathbf{1.95\cdot 10^{-2}} \pm 1.66 \cdot 10^{-3}$} & \multicolumn{1}{c|}{$\mathbf{1.14\cdot 10^{-2}} \pm 1.38  \cdot 10^{-3}$} & \multicolumn{1}{c|}{$\mathbf{3.70 \cdot 10^{-3}} \pm 8.55 \cdot 10^{-4}$} \\
\multicolumn{1}{|l|}{} & \multicolumn{1}{c|}{MNN} & \multicolumn{1}{c|}{$2.48\cdot 10^{-2} \pm 6.55 \cdot 10^{-3}$} & \multicolumn{1}{c|}{$2.52\cdot 10^{-2} \pm 1.20 \cdot 10^{-2}$} & \multicolumn{1}{c|}{$8.16 \cdot 10^{-2} \pm 1.87 \cdot 10^{-1}$} \\
\multicolumn{1}{|l|}{} & \multicolumn{1}{c|}{\textcolor{black}{MLP}} & \multicolumn{1}{c|}{$2.01\cdot 10^{-2} \pm 1.66 \cdot 10^{-3}$} & \multicolumn{1}{c|}{$1.19\cdot 10^{-2} \pm 1.24 \cdot 10^{-3}$} & \multicolumn{1}{c|}{$3.81\cdot 10^{-3} \pm 8.46 \cdot 10^{-4}$} \\
\hline
\end{tabular}}
\caption{MSE on the wind field reconstruction task}
\label{table:results_recon}
\end{table*}
\begin{table*}[]
\centering
\scalebox{.75}{\begin{tabular}{|l|cccc|}
\hline
& & $T_f = 20$ & $T_f = 50$ & $T_f = 80$ \\
\hline
\multicolumn{1}{|l|}{\multirow{3}{*}{$\textrm{E}\{n\}=100$}} & \multicolumn{1}{c|}{DD-TNN} & \multicolumn{1}{c|}{$\mathbf{8.39\cdot 10^{-2}} \pm 1.62\cdot 10^{-2}$} & \multicolumn{1}{c|}{$1.20\cdot 10^{-1} \pm 7.13 \cdot 10^{-2}$} & \multicolumn{1}{c|}{$1.36\cdot 10^{-1} \pm 5.05 \cdot 10^{-2}$} \\
\multicolumn{1}{|l|}{} & \multicolumn{1}{c|}{MNN} & \multicolumn{1}{c|}{$3.76\cdot 10^{-1} \pm 2.74 \cdot 10^{-1}$} & \multicolumn{1}{c|}{$6.18\cdot 10^{-1} \pm 2.09 \cdot 10^{-1}$} & \multicolumn{1}{c|}{$9.63\cdot 10^{-1} \pm 1.55 \cdot 10^{-2}$} \\
\multicolumn{1}{|l|}{} & \multicolumn{1}{c|}{\textcolor{black}{RNN}} & \multicolumn{1}{c|}{$8.89\cdot 10^{-2} \pm 3.06 \cdot 10^{-2}$} & \multicolumn{1}{c|}{$\mathbf{8.69\cdot 10^{-2}} \pm 3.19 \cdot 10^{-2}$} & \multicolumn{1}{c|}{$\mathbf{7.24\cdot 10^{-2}} \pm 2.63 \cdot 10^{-2}$} \\
\hline
\multicolumn{1}{|l|}{\multirow{3}{*}{$\textrm{E}\{n\}=200$}} & \multicolumn{1}{c|}{DD-TNN} & \multicolumn{1}{c|}{$\mathbf{8.14\cdot 10^{-2}} \pm 3.58\cdot 10^{-2}$} & \multicolumn{1}{c|}{$\mathbf{7.03\cdot 10^{-2}} \pm 1.40 \cdot 10^{-2}$} & \multicolumn{1}{c|}{$1.61\cdot 10^{-1} \pm 7.38 \cdot 10^{-2}$} \\
\multicolumn{1}{|l|}{} & \multicolumn{1}{c|}{MNN} & \multicolumn{1}{c|}{$3.49\cdot 10^{-1} \pm 2.03 \cdot 10^{-1}$} & \multicolumn{1}{c|}{$6.70\cdot 10^{-1} \pm 2.30 \cdot 10^{-1}$} & \multicolumn{1}{c|}{$5.03\cdot 10^{-1} \pm 1.39 \cdot 10^{-1}$} \\
\multicolumn{1}{|l|}{} & \multicolumn{1}{c|}{\textcolor{black}{RNN}} & \multicolumn{1}{c|}{$9.78\cdot 10^{-2} \pm 3.58 \cdot 10^{-2}$} & \multicolumn{1}{c|}{$7.54\cdot 10^{-2} \pm 4.92 \cdot 10^{-2}$} & \multicolumn{1}{c|}{$\mathbf{1.20\cdot 10^{-1}} \pm 5.99 \cdot 10^{-2}$} \\
\hline
\multicolumn{1}{|l|}{\multirow{3}{*}{$\textrm{E}\{n\}=300$}} & \multicolumn{1}{c|}{DD-TNN} & \multicolumn{1}{c|}{$\mathbf{6.43\cdot 10^{-2}} \pm 1.13\cdot 10^{-2}$} & \multicolumn{1}{c|}{$\mathbf{7.03\cdot 10^{-2}} \pm 3.16 \cdot 10^{-2}$} & \multicolumn{1}{c|}{$2.49\cdot 10^{-1} \pm 1.74 \cdot 10^{-1}$} \\
\multicolumn{1}{|l|}{} & \multicolumn{1}{c|}{MNN} & \multicolumn{1}{c|}{$6.69\cdot 10^{-1} \pm 2.12 \cdot 10^{-1}$} & \multicolumn{1}{c|}{$6.09\cdot 10^{-1} \pm 3.66 \cdot 10^{-1}$} & \multicolumn{1}{c|}{$8.83\cdot 10^{-1} \pm 8.60 \cdot 10^{-2}$} \\
\multicolumn{1}{|l|}{} & \multicolumn{1}{c|}{\textcolor{black}{RNN}} & \multicolumn{1}{c|}{$7.95\cdot 10^{-2} \pm 3.43 \cdot 10^{-2}$} & \multicolumn{1}{c|}{$8.68\cdot 10^{-2} \pm 4.06 \cdot 10^{-2}$} & \multicolumn{1}{c|}{$\mathbf{1.34\cdot 10^{-1}} \pm 4.66 \cdot 10^{-2}$} \\
\hline
\multicolumn{1}{|l|}{\multirow{3}{*}{$\textrm{E}\{n\}=400$}} & \multicolumn{1}{c|}{DD-TNN} & \multicolumn{1}{c|}{$8.93\cdot 10^{-2} \pm 2.78 \cdot 10^{-2}$} & \multicolumn{1}{c|}{$\mathbf{8.47\cdot 10^{-2}} \pm 1.67 \cdot 10^{-2}$} & \multicolumn{1}{c|}{$1.34 \cdot 10^{-1} \pm 4.35 \cdot 10^{-2}$} \\
\multicolumn{1}{|l|}{} & \multicolumn{1}{c|}{MNN} & \multicolumn{1}{c|}{$4.06\cdot 10^{-1} \pm 2.47 \cdot 10^{-1}$} & \multicolumn{1}{c|}{$7.50\cdot 10^{-1} \pm 2.86 \cdot 10^{-1}$} & \multicolumn{1}{c|}{$2.35 \cdot 10^{-1} \pm 9.49 \cdot 10^{-2}$} \\
\multicolumn{1}{|l|}{} & \multicolumn{1}{c|}{\textcolor{black}{RNN}} & \multicolumn{1}{c|}{$\mathbf{6.29\cdot 10^{-2}} \pm 2.66 \cdot 10^{-2}$} & \multicolumn{1}{c|}{$1.25\cdot 10^{-1} \pm 3.89 \cdot 10^{-2}$} & \multicolumn{1}{c|}{$\mathbf{5.19\cdot 10^{-2}} \pm 3.38 \cdot 10^{-2}$} \\
\hline
\end{tabular}}
\caption{MSE on the wind field forecasting task}
\label{table:results_forec}
\end{table*}
\begin{figure*}[t]
     \centering
    \hspace{-.5cm}
     \begin{subfigure}[b]{0.4\textwidth}
         \centering
         \includegraphics[scale = 0.145]{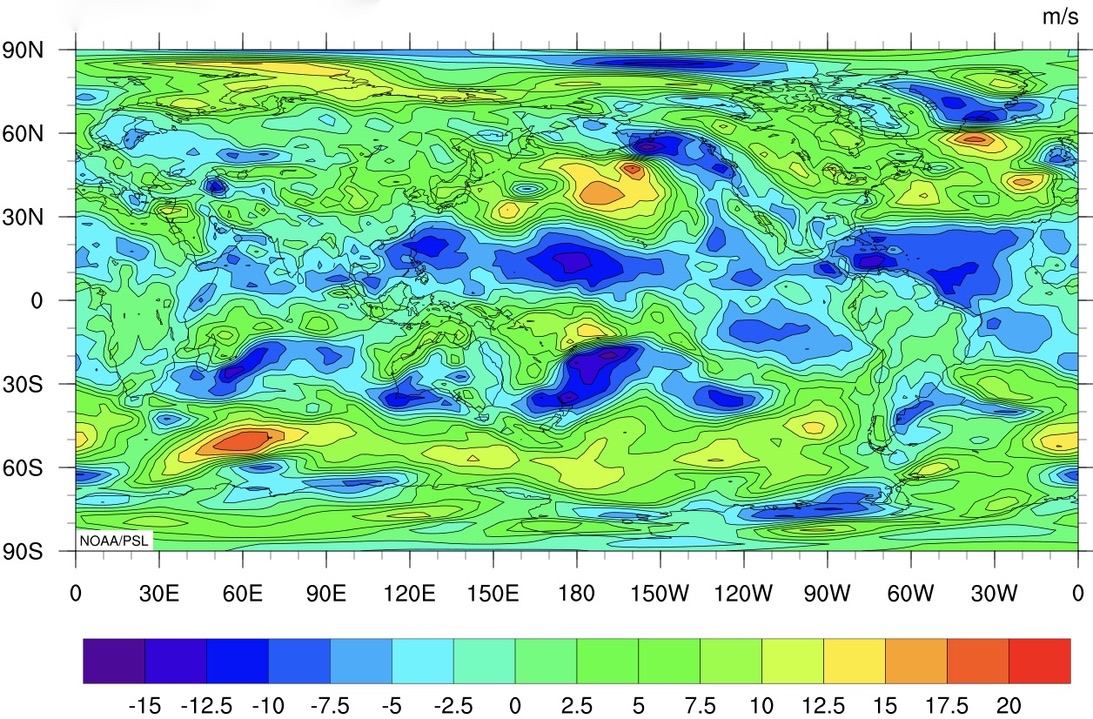}
         \caption{Zonal Wind}
         \label{fig:synth}
     \end{subfigure}
     \hspace{0.5cm}
     \begin{subfigure}[b]{0.4\textwidth}
         \centering
         \includegraphics[scale = 0.145]{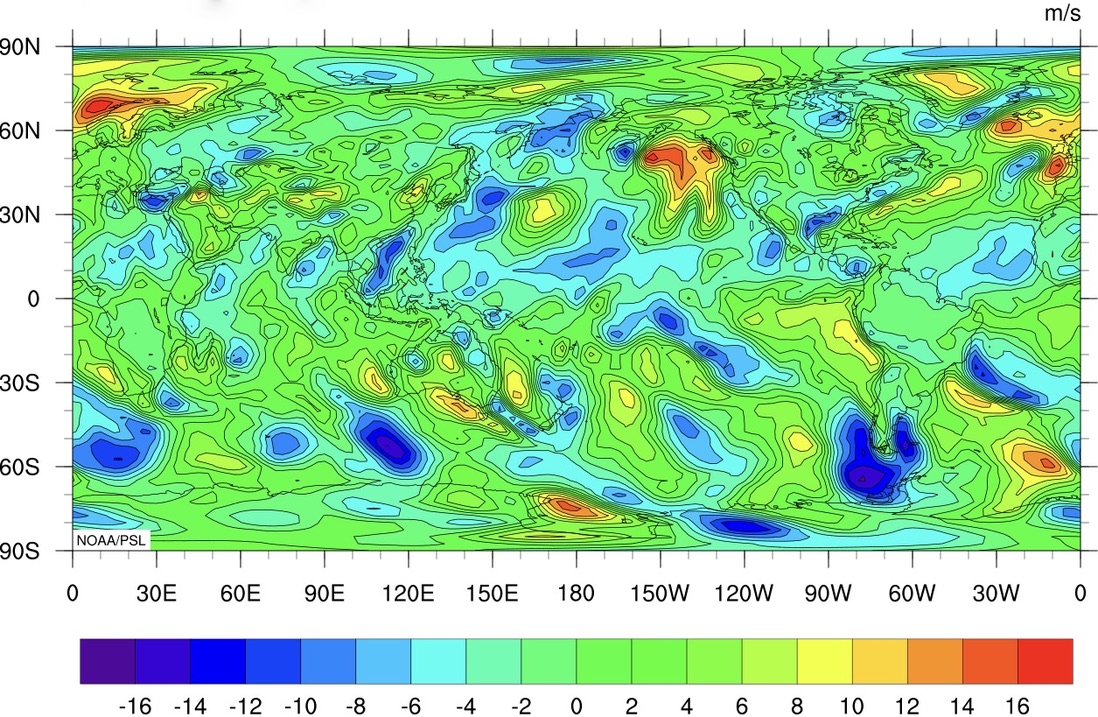}
         \caption{Meridional Wind}
         \label{fig:ocean}
     \end{subfigure}
        \caption{Visualization of Earth wind field on 1st of January 2016 (a) Zonal component. (b) Meridional component.}
        \label{fig:maps}
\end{figure*}
\vspace{-.2cm}
\subsection{Wind Field Reconstruction}\label{sec:wind_rec}
We design a reconstruction task on real-world data. We use daily average measurements (the tangent bundle signal) of Earth surface wind field collected by NCEP/NCAR\footnote{https://psl.noaa.gov/data/gridded/data.ncep.reanalysis.html}; in particular, we use the data corresponding to the wind field of the 1st of January 2016, consisting of regularly spaced observations covering the whole Earth surface. The observations are localized in terms of latitude and longitude, thus we convert them in 3-dimensional coordinates by using the canonical spherical approximation for the Earth with nominal radius $R=6356.8$. The wind field is a 2-dimensional tangent vector field made of a zonal component, following the local parallel of latitude, and a meridional component, following the local meridian of longitude. A visualization of the wind field is shown in Fig. \ref{fig:maps} (figures taken from the official data repository). We preprocess the data by scaling the observations to be in the range $[-1,1]$. We first randomly sample $n$ points to obtain the sampling set $\mathcal{X}$, the cellular sheaf $\tm_n$, and the Sheaf Laplacian $\Delta_n$ again with $\epsilon_{\textrm{PCA}}=0.8$ and $\epsilon_n = 0.5$; at this point, we mask $\widetilde{n}<n$ of these points, we collect them in a set $\widetilde{\mathcal{X}}^C \subset \mathcal{X}$, and we aim to infer their corresponding measurements exploiting the remaining available $n-\widetilde{n}$ measurements, collected in the set $\widetilde{\mathcal{X}} \subset \mathcal{X}$. This reconstruction problem can be equivalently seen as a semi-supervised regression problem. To tackle it, we first organize the data corresponding to the point in $\mathcal{X}$ in a matrix $\mathbf{F}_n \in \mathbb{R}^{n \times 2}$, where the first column collects the zonal components and the second column collects the meridional components. At this point, we build the matrix $\widetilde{\mathbf{F}}_n\in \mathbb{R}^{n \times 2}$, that is a copy of $\mathbf{F}$ except for the rows of $\mathbf{F}$ corresponding to the masked points in $\widetilde{\mathcal{X}}^C$, replaced with the mean of the measurements of the available points in  $\widetilde{\mathcal{X}}$. We then vectorize $\widetilde{\mathbf{F}}_n$ to obtain $\widetilde{\mathbf{f}}_{n} \in \mathbb{R}^{2n}$, the input tangent bundle signal. We now exploit the same DD-TNN architecture from Section \ref{torus_den}, with the same hyperparameters, to perform the reconstruction task by training it to minimize the reconstruction square error
\begin{equation}
\sum_{i \in \widetilde{\mathcal{X}}}\|\widetilde{\mathbf{f}}_n(i) - \mathbf{f}_{n}^o(i)\|^2
\end{equation}
between the available measurements $\mathbf{f}_n(i)$ and the output of the network corresponding to them $\mathbf{f}_{n}^o(i)$, $i \in \widetilde{\mathcal{X}}$. Again, we compare our architecture with the same MNN from Section \ref{torus_den}, to which we give as input the matrix $\widetilde{\mathbf{F}}$ and we train it to minimize $\sum_{i \in \widetilde{\mathcal{X}}}\|\widetilde{\F}_n(i) - \mathbf{F}_{n}^o(i)\|^2$, where  $\mathbf{F}_{n}^o$ is the network output and $\widetilde{\F}_n(i)$ indicates the $i-$th row of $\widetilde{\F}_n(i)$; being $\widetilde{\f}_n$ the vectorization of $\widetilde{\F}_n$, also in this case it is trivial to check the equivalence of the two MSEs. As evaluation metric, we use the reconstruction MSE on the measurements corresponding to the masked nodes $\frac{1}{n}\sum_{i \in \widetilde{\mathcal{X}}^C}\|\mathbf{f}_n(i) - \mathbf{f}_{n}^o(i)\|^2$.
In Table \ref{table:results_recon} we evaluate TNNs and MNNs for four different expected sample sizes ($\textrm{E}\{n\} = 100$, $\textrm{E}\{n\} = 200$, $\textrm{E}\{n\} = 300$, and $\textrm{E}\{n\}=400$), for three different masking probabilities ($\textrm{E}\{\widetilde{n}\} = 0.5n$, $\textrm{E}\{\widetilde{n}\} = 0.3n$, and $\textrm{E}\{\widetilde{n}\} = 0.1n$) per each of them (the probability of a node to being masked). As the reader can notice, TNNs are always able to perform better than MNNs and MLPs, keeping the performance stable with the number of samples and, of course, improving with more observations available.
\vspace{-.3cm}
\subsection{Wind Field Forecasting with Recurrent TNNs}
\vspace{-.1cm}
We design a forecasting task on the same wind field data from Section \ref{sec:wind_rec}. In particular, we use daily observation corresponding to the wind field from the 1st of January 2016 to 7 September 2016 to train the model and we use observations from the 1st of January 2017 to 7 September 2017 to test it. We, again,  randomly sample $n$ points to obtain the sampling set $\mathcal{X}$, the cellular sheaf $\tm_n$, and the Sheaf Laplacian $\Delta_n$; at this point, we organize the data corresponding to the sampled point in $\mathcal{X}$ in a sequence $\{\mathbf{F}_{n,t}\}_t$ indexed by time $t$ (daily interval), with each $\mathbf{F}_{n,t} \in \mathbb{R}^{n \times 2}$. As in Section \ref{sec:wind_rec}, we  vectorize $\{\mathbf{F}_{n,t}\}_t$  to obtain $\{\mathbf{f}_{n,t}\}_t$, the input tangent bundle signals, with each   $\mathbf{f}_{n,t}\in \mathbb{R}^{2n}$. We now introduce a hyperparameter $T_f>0$ representing the length of the predictive time window of the model, i.e., given in input a subsequence $\{\mathbf{f}_{n,t}\}_{t=T_s}^{t =T_s+T_f}$ starting at time $T_s$ of length $T_f$ , the model outputs a sequence $\{\mathbf{f}^{o}_{n,t}\}_{t=1}^{t =T_f}$ of length $T_f$  aiming at estimating the next $T_f$ element $\{\mathbf{f}_{n,t}\}_{t=T_s+T_f+1}^{t =T_s+2T_f+1}$ of the input sequence. To do so, we introduce a recurrent version of the proposed DD-TNNs, which, to the best of our knowledge, is also the first recurrent architecture working on cellular sheaves. The building block of the proposed recurrent architecture is a layer made of three components: a tangent bundle filter processing the current sequence element $\mathbf{f}_{n,t}$, a tangent bundle filter processing the current hidden state $\mathbf{z}_{t-1}$, i.e., the output of the layer computed on the previous sequence element, and a pointwise non-linearity. Formally, the layer reads as:
\begin{equation}\label{rtnn}
    \mathbf{z}_{t} = \sigma\Bigg(\sum_{k = 0}^{K-1}h_{k}\big(e^{\Delta_n}\big)^k\mathbf{f}_{n,t} + \sum_{k = 0}^{K-1}w_{k}\big(e^{\Delta_n}\big)^k\mathbf{z}_{t-1}\Bigg),
\end{equation}
with $t = T_s,..., T_s + T_f$, and $\mathbf{z}_{0}=\mathbf{0}$. To obtain the required estimates,  we can set $\{\mathbf{f}^{o}_{n,t}\}_{t=1}^{t =T_f}=\{\mathbf{z}_t\}_{t=1}^{t =T_f}$. This architecture can be used also in a Multi-Layer fashion: in this case, at layer $l$ and at time $t$, the first filter takes $\mathbf{z}_{l-1,t}$ (the current time $t$ hidden state of the previous layer $l-1$) as input, and the second filter takes $\mathbf{z}_{l,t-1}$ (the previous time $t-1$ hidden state of the current layer $l$) as input. Therefore, the resulting $L-$layers architecture is:
\begin{equation}\label{rtnn_ml}
    \mathbf{z}_{l,t} = \sigma\Bigg(\sum_{k = 0}^{K-1}h_{k,l}\big(e^{\Delta_n}\big)^k\mathbf{z}_{l-1,t} + \sum_{k = 0}^{K-1}w_{k,l}\big(e^{\Delta_n}\big)^k\mathbf{z}_{l,t-1}\Bigg),
\end{equation}
with $l=1,...,L$, $t = T_s,..., T_s + T_f$, and $\mathbf{z}_{0,t}=\mathbf{f}_{n,t}$. In this case, to obtain the required estimates,  we can set $\{\mathbf{f}^{o}_{n,t}\}_{t=1}^{t =T_f}=\{\mathbf{z}_{L,t}\}_{t=1}^{t =T_f}$. For the wind field forecasting task, the training set is made of all the possible $m=250-2T_f$ subsequences of length $2T_f$ of the 2016 data, we use a 3-layers Recurrent DD-TNN with $K=2$ and $\textrm{Tanh}$ non-linearities, and we train it to minimize the square error $\sum_{t =1}^m\sum_{\widetilde{t} =t}^{t+T_f}\|\mathbf{f}_{n,\widetilde{t}} - \mathbf{f}^{o}_{n,\widetilde{t}-t+1}\|_2^2$. To have a fair comparison, we set up the corresponding recurrent version of MNNs (RMNNs, a recurrent graph neural network) with the same structure, same hyperparameters, same loss but with inputs $\{\mathbf{F}_{n,t}\}_t$. As evaluation metric, we compute the MSE on the 2017 data after training. In Table \ref{table:results_forec} we evaluate RTNNs and RMNNs for four different expected sample sizes ($\textrm{E}\{n\} = 100$, $\textrm{E}\{n\} = 200$, $\textrm{E}\{n\} = 300$, and $\textrm{E}\{n\}=400$), and for three different time window lengths ($T_f = 20$, $T_f = 50$, and $T_f = 80$) per each of them. Also in this case, the bundle "awareness" of (recurrent) TNNs allows to reach significantly better results in all the tested scenarios w.r.t. (recurrent) MNNs, \textcolor{black}{outperforming RNNs too in most of the cases except for the long-term predictions cases ($T_f = 80$),  probably due to the absence of gating or memory mechanisms, useful also to improve training}. Moreover, in all the experiments, TNNs have fewer parameters than MNNs, due to the different organization of the data in the input layer. 
\begin{table}
\centering
\textcolor{black}{\scalebox{.75}{\begin{tabular}{|l|l|l|l|}
\hline
       & $\textrm{Id}()$    & $\textrm{Tanh}()$         \\
\hline
DD-TNN & $76.2 \% \pm 4.8$ & $87.5 \% \pm 0.9$      \\
\hline
MNN    & $56.3 \% \pm  5.3$ & $97.7 \% \pm 0.3$      \\
\hline
\textcolor{black}{3D-CNN}    & $\mathbf{86.3 \% \pm  2.3}$ & $\mathbf{98.5 \% \pm 0.6}$      \\
\hline
\end{tabular}}}
\caption{\textcolor{black}{Accuracy on the manifold classification task}}
\label{table:results_class}
\end{table}
\vspace{-.7cm}
\textcolor{black}{\subsection{Manifold Classification}\label{sec:man_class}
\begin{figure}[t]
         \centering
         \includegraphics[scale = .28]{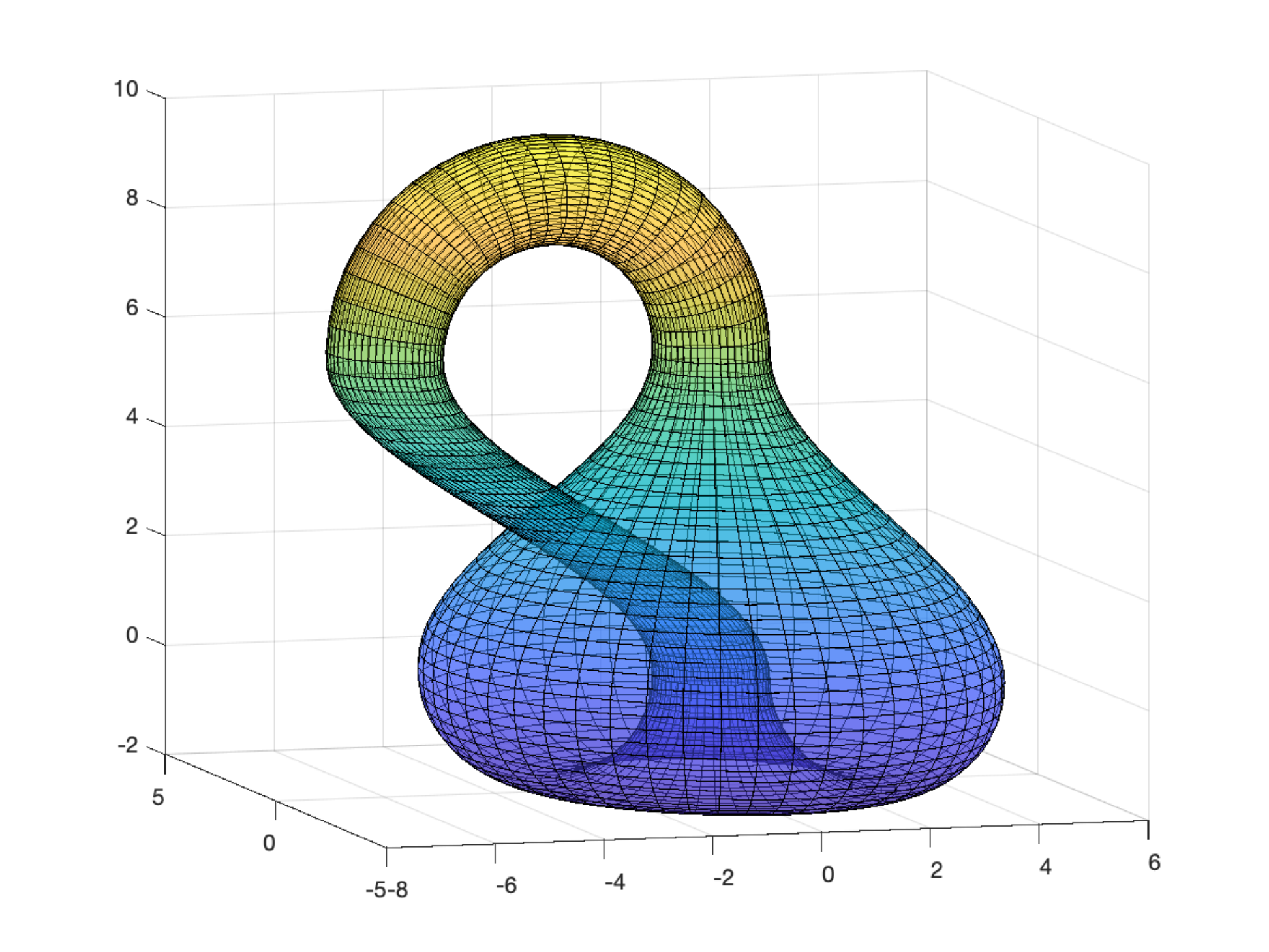}
         \label{fig:klein}
        \caption{\textcolor{black}{(b) An immersed Klein bottle.}}
        \label{fig:manifolds}
\end{figure}
We design an ill-conditioned binary manifold classification task. The goal is to assess whether TNNs are able to learn to distinguish between a torus and an immersed Klein bottle, given uninformative features, i.e. a constant vector field. The Klein bottle is a 2-dimensional non-orientable manifold (i.e.~with no ``inside'' or ``outside"). \textcolor{black}{The Klein bottle has to be properly embedded in $\mathbb{R}^4$ since it must pass through itself without the presence of a hole.  However, it can be immersed in $\mathbb{R}^3$, as depicted in Fig. depicted in Fig. \ref{fig:manifolds}(b). Working on the immersion makes the classification task ill-conditioned, because an immersed Klein Bottle is not even a proper manifold.}  The choice of giving as input uninformative features is made in order to evaluate if the network is able to solely leverage the geometric information contained in the approximated Connection Laplacian (the Sheaf Laplacian) to infer the manifold, resembling \cite{GIN_xu2019powerful,ribeiro2022aremore}. We opted to employ the immersed Klein bottle as one of the target manifolds to heuristically evaluate how our method tackles ill-conditioned tasks that do not match the theoretical justification of our architecture.  We create datasets of $2000$ data points, where each datapoint is computed with the following steps: i) choose if sampling the torus or the Klein bottle tossing a fair coin; ii) uniformly sample the manifold on $\textrm{E}\{n\}=100$ points and compute the corresponding cellular sheaf $\tm_n$, and  Sheaf Laplacian $\Delta_n$ (both the manifolds are normalized to be contained in a cube of unitary side length, such that the scale is not a discriminating feature); iii) associate to each sampled point the projection via the $\Oi$s of the constant vector field given by $d\i\F(x,y,z)=(1,1,1) \in 
\mathbb{R}^3$. We exploit the same DD-TNN architecture from Section \ref{torus_den}, with the same hyperparameters, plus a final 2-layers MLP classifier with a $\textrm{ReLU}()$ non-linearity after the first layer and a $\textrm{Softmax}()$ non-linearity after the second layer. We train the architecture to minimize the usual binary cross-entropy. We compare our architecture with the same MNN from Section \ref{torus_den}, where obviously the Sheaf Laplacian is replaced by the graph Laplacian. Unlike the experiments of the previous section, we found that in this case the employed non-linearity impacts the classification accuracy; in particular, we observed changes when the "non-linearity" is the identity function $\textrm{Id}()$ (thus when a cascade of discretized tangent bundle filters is used) or the $\textrm{Tanh}()$. For this reason we evaluate TNNs and MNNs using both $\textrm{Id}()$ and $\textrm{Tanh}()$. \textcolor{black}{Moreover, we also report the results of a simple 3D-CNN. To feed the data in the 3D-CNN, we follow the approach from \cite{qi2017pointnet}, i.e. we first convert each data point to the volumetric representation as an occupancy grid with resolution 6 × 6 × 6. The choice of the hyperparameters is made to keep the number of learnable parameters similar to DD-TNN and MNN.} The results are averaged over 8 realizations of the datasets, and per each of them the training dataset is obtained by sampling $80\%$ of the datapoints, and the test set is obtained with the remaining $20\%$. \textcolor{black}{As the reader can notice in Table \ref{table:results_class}, the 3D-CNN performs better than both MNN and DD-TNN. This is something that we could expect: as described above, this task is ill-conditioned for DD-TNN and MNN, and it is trivially easier to solve for an architecture designed for point clouds, the 3D-CNN w.r.t. architectures based on manifold diffusion operators. However, the performance of DD-TNN is still competitive even if the setting is disadvantaged.} Moreover, TNNs significantly perform better than MNNs when $\textrm{Id}()$ is employed, i.e. tangent bundle filters significantly perform better than manifold filters, while MNNs perform better than TNNs when an actual non-linearity, the $\textrm{Tanh}()$, is introduced.}
\vspace{-.1cm}
\section{Conclusions}\label{sec:conlcusion}
In this work, we introduced Tangent Bundle Filters and Tangent Bundle Neural Networks (TNNs), novel continuous architectures operating on tangent bundle signals, i.e. manifold vector fields. We made TNNs implementable by discretization in space and time domains, showing that their discrete counterpart is a principled variant of Sheaf Neural Networks. We proved that discretized TNNs asymptotically converge to their continuous counterparts, and we assessed the performance of TNNs on both synthetic and real dat{}a. This work gives a multifaceted contribution: on the methodological side, it is the first work to introduce a signal processing framework for signals defined on tangent bundles of Riemann manifolds via the Connection Laplacian; on the theoretical side, the presented discretization procedure and convergence result explicitly link the manifold domain with cellular sheaves, formalizing intuitions presented in works like \cite{barbero2022sheafnnconn}. \textcolor{black}{In future work, we will investigate more general classes of cellular sheaves that approximate unions of manifolds (perhaps representing multiple classes) or, more generally, stratified spaces \cite{stolz2020geometric,nanda2020local}. We believe our perspective on tangent bundle neural networks could shed further light on challenging problems in graph neural networks such as heterophily \cite{bodnar2022sheafdiff}, over-squashing \cite{digiovanni2023oversquashing}, or transferability \cite{ levie2021transferability}. Finally,we plan to tackle more sophisticated tasks with our proposed architectures.}
\vspace{-.1cm}
\appendix
\section{Appendix}
\vspace{-.1cm}
\subsection{Proof of Theorem 1}\label{ap:proofth1}
\noindent\textbf{\textit{Proof.}}  We define an inner product for sheaf signals $\mathbf{f}$ and $\mathbf{u}$ on a general cellular sheaf $\tm_n$ as
\begin{align}
\label{emp_metr_sheaf}
\langle \mathbf{f}, \mathbf{u} \rangle_{\tm_n}   & = \frac{1}{n}\sum_{i = 1}^n  \f_{n}(x_i) \dotp \mathbf{u}_{n}(x_i) ,
\end{align}
and the induced norm $||\mathbf{f}||^2_{\tm_n} = \langle \mathbf{f}, \mathbf{f} \rangle_{\tm_n}$.
Assuming that the points in $\mathcal{X}$ are sampled i.i.d. from the uniform probability measure $\mu$ given by the metric on $\M$ and that $\tm_n$ is built as in Section 5, the inner product in \eqref{emp_metr_sheaf} is equivalent to the following inner product for tangent bundle signals $\F$ and $\U$:
\begin{align}
\label{emp_metr}
\langle \F, \U \rangle_{\tm_n}   &= \int_{\M} \langle \F(x), \U(x) \rangle_{\tmx} \textrm{d}\mu_n(x) \nonumber \\
&= \frac{1}{n}\sum_{i = 1}^n \langle \F(x_i), \mathbf{U}(x_i) \rangle_{\mathcal{T}_{x_i}\M},
\end{align}
and the induced norm $||\F||^2_{\tm_n} = \langle \F, \F \rangle_{\tm_n}$, where $\mu_n = \frac{1}{n}\sum_{i=1}^n\delta_{x_i}$ is the empirical measure corresponding to $\mu$. Indeed, from \eqref{riemann_metric} and due to the orthogonality of the transformations $\Oi$ in Section 5, \eqref{emp_metr} can be  rewritten as
\begin{align}
\label{emp_metr_versions}
\langle \F, \U \rangle_{\tm_n}&= \frac{1}{n}\sum_{i = 1}^n  d\i\F(x_i) \dotp   d\i\mathbf{U}(x_i) \nonumber \\
&= \frac{1}{n}\sum_{i = 1}^n \Oi^T d\i\F(x_i) \dotp  \Oi^T d\i\mathbf{U}(x_i) \nonumber \\
&= \frac{1}{n}\sum_{i = 1}^n \f_{n}(x_i) \dotp \mathbf{u}_{n}(x_i)  = \langle \mathbf{f}_n, \mathbf{u}_n \rangle_{\tm_n}
\end{align}
where $\f_n = \samp_n^{\mathcal{X}}\F$ and $\mathbf{u}_n = \samp_n^{\mathcal{X}}\U$, respectively. We denote with $\ltmn$ the space of tangent bundle signals w.r.t. the empirical measure $\mu_n$ (or, equivalently, the space of sheaf signals w.r.t the norm induced by \eqref{emp_metr_sheaf}). In the following, we will denote the norm $||\cdot||_{\tm_n}$ with $||\cdot||$  when there is no risk of confusion. In \cite{singer2017spectral}, the spectral convergence of the constructed Sheaf Laplacian in \eqref{sheaf_laplacian} to the Connection Laplacian of the underlying manifold has been proved, and we exploit that result for proving the following proposition.

\noindent\textbf{\textit{Proposition 3. (Consequence of Theorem 6.3 \cite{singer2017spectral})}}
 Let $\mathcal{X}=\{x_1,\dots,x_n\}\subset \mathbb{R}^p$ be a set of $n$ i.i.d. sampled points from measure $\mu$ over $\M \subset \mathbb{R}^p$. Let $\tm_n$ be a cellular sheaf built from $\mathcal{X}$ as explained in Section 5, \textcolor{black}{with $\hd=d$ and $0<\epsilon_n \leq \min\{\kappa^{-1}, \iota\}$}. Let $\Delta_n$ be the Sheaf Laplacian of $\tm_n$ and $\Delta$ be the Connection Laplacian operator of $\M$. Let $\lambda_{i}^n$ be the $i$-th eigenvalue of $\Delta_n$ and $\Phii^n$ the corresponding eigenvector. Let $\lambda_i$ be the $i$-th eigenvalue of $\Delta$ and $\Phii$  the corresponding eigenvector field of $\Delta$, respectively. \textcolor{black}{Then there exists a sequence of scales $\epsilon_n \rightarrow 0$ as $n \rightarrow \infty$ such that:}
\begin{equation}
\label{eqn:convergence_spectrum}
    \lim_{n\rightarrow \infty } \lambda_i^n = \lambda_i, \quad\lim_{n\rightarrow \infty} \|\Phii^{n} -  \samp_n^{\mathcal{X}}\Phii\|_{\tm_n}=0,
\end{equation}
where the limits are taken in probability.

\noindent\textbf{\textit{Proof.}} This proposition is a consequence of Theorem 6.3 in \cite{chung1997spectral}. Indeed, we rely on the operator introduced \textcolor{black}{in Definition 6.1 of \cite{singer2017spectral}} with $\alpha=1$ ($h_n$ is our $\epsilon_n$), here denoted as $\Xi:\ltm \rightarrow \ltm$, and on the operator $\widetilde{\Xi} = \epsilon_n^{-1}\big(\Xi - \textrm{id}\big)$, where $\textrm{id}$ is the identity mapping. It is straightforward to check:
\begin{equation}
    \label{omega_delta_equiv}
    \widetilde{\Xi}\F(x_j) =  d\i^{-1}\Oj\big(\Delta_n\samp_n^{\mathcal{X}}\F\big)(x_j),
\end{equation}
for $j = 1,\dots,n$. We now show that the eigenvectors sampled on $\mathcal{X}$ and eigenvalues of $\widetilde{\Xi}$  correspond to the eigenvectors and eigenvalues of $\Delta_n$. Let us denote the the $i$-th eigenvector and eigenvalue of $\widetilde{\Xi}$ with $\widetilde{\boldsymbol{\phi}}^n_{i}$ and $-\widetilde{\lambda}^n_{i}$, respectively. We have:
\begin{align}
    \label{eigen_conv}
    \widetilde{\Xi}\eigGammai(x_j) = -\eivGammai\eigGammai(x_j)=d\i^{-1}\Oj\big(\Delta_n\samp_n^{\mathcal{X}}\eigGammai\big)(x_j)
\end{align}
If we apply the mapping $i$ to the last two equalities of \eqref{eigen_conv} and we exploit the orthoghonality of $\Oj$, we obtain:
\begin{align}
\label{eig_gamma_delt}
    \big(\Delta_n\samp_n^{\mathcal{X}}\eigGammai\big)(x_j) = -\eivGammai\Oj^T d\i\eigGammai &= -\eivGammai\samp_n^{\mathcal{X}}\eigGammai(x_j)
\end{align}
where the second equality applies the definition of $\samp_n^{\mathcal{X}}$ in \eqref{sampler}. Therefore, we have:
\begin{align}
\label{eig_gamma_delta}
    \lambda_i^n = \eivGammai, \quad \Phii^n(x_j) = \samp_n^{\mathcal{X}}\eigGammai(x_j),
\end{align}
$j= 1,\dots,n$. At this point, we can recall Theorem 6.3 in \cite{singer2017spectral}, that, in the setting of our Theorem 1, states that \textcolor{black}{there exists a sequence of scales $\epsilon_n \rightarrow 0$ as $n \rightarrow \infty$ such that:}
\begin{equation}
    \label{spect_VDM}
    \lim_{n\rightarrow \infty } \widetilde{\lambda}_i^n = \lambda_i, \quad\lim_{n\rightarrow \infty} \|\eigGammai -  \Phii\|_{\tm}=0,
\end{equation}
with the limit taken in probability, $j= 1,\dots,n$. 
Injecting the empirical measure in \eqref{spect_VDM} and exploiting the results in \eqref{emp_metr_versions} and \eqref{eig_gamma_delta}, we obtain:
\begin{align}
\label{norm_tm}
    &\|\eigGammai-  \Phii\|_{\tm_n} = \|\Phii^n- \samp_n^{\mathcal{X}}\Phii\|_{\tm_n}
\end{align}
The results in \eqref{spect_VDM} and \eqref{norm_tm}
and  the a.s. convergence of the empirical measure $\mu_n$ to the measure $\mu$ conclude the proof.

For the sake of clarity, in the following we will drop the dependence on the NNs output index $u$; from the definitions of TNNs in \eqref{tnn_layer} and D-TNNS in \eqref{dt_tnn_layer}, we can thus write:
 \begin{align}
    \nonumber  \|\bPsi\big(\mathcal{H}, \Delta_n, \samp_n^{\mathcal{X}}\F\big) - \samp_n^{\mathcal{X}}\bPsi\big(\mathcal{H}, \Delta,\F\big)\|&= \left\| \bbx_{n,L}-\samp_n^{\mathcal{X}}\F_L \right\|.
 \end{align}
 Further explicating the layers definitions, at layer $l$ we have: 
 \begin{align}
   \nonumber  &\left\| \bbx_{n,l}- \samp_n^{\mathcal{X}} \F_l \right\|\\
     &=\left\| \sigma\left(\sum_{q=1}^{F_{l-1}} \bbh_l^{q}(\Delta_n) \bbx_{n,l-1}^q \right) -\samp_n^{\mathcal{X}} \sigma\left(\sum_{q=1}^{F_{l-1}} \bbh_l^{q}(\Delta) \F_{l-1}^q\right) \right\|
 \end{align}
 with $\bbx_{n,0}^q=\samp_n^{\mathcal{X}} \F^q$ for $q=1,\dots,F_0$. Exploiting the normalized point-wise Lipschitz continuity of the non-linearities (\textbf{A3}) and the linearity of the sampling operator $\samp_n^{\mathcal{X}}$, we have:
  \begin{align}
  \label{proof_1}
    \| \bbx_{n,l} - \samp_n^{\mathcal{X}} \F_l  \| &\leq \Bigg\|  \sum_{q=1}^{F_{l-1}} \bbh_l^{q}(\Delta_n) \bbx_{n,l-1}^q - \samp_n^{\mathcal{X}} \sum_{q=1}^{F_{l-1}} \bbh_l^{q}(\Delta)  \F_{l-1}^q\Bigg\| \nonumber\\
    & \leq \sum_{q=1}^{F_{l-1}} \left\|    \bbh_l^{q}(\Delta_n) \bbx_{n,l-1}^q    - \samp_n^{\mathcal{X}}   \bbh_l^{q}(\Delta)  \F_{l-1}^q\right\|
 \end{align}
 The difference term in the last LHS of \eqref{proof_1} can be further decomposed for every $q=1,\dots,F_{l-1}$ as
\begin{align}\label{proof_2}
   \nonumber   \|    \bbh_l^{q}(\Delta_n) & \bbx_{n,l-1}^q    - \samp_n^{\mathcal{X}}   \bbh_l^{q}(\Delta)  \F_{l-1}^q \| 
   \\ \nonumber&\leq \|
\bbh_l^{q}(\Delta_n) \bbx_{n,l-1}^q  - \bbh_l^{q}(\Delta_n) \samp_n^{\mathcal{X}} \F_{l-1}^q \\ &\qquad +\bbh_l^{q}(\Delta_n) \samp_n^{\mathcal{X}} \F_{l-1}^q  - \samp_n^{\mathcal{X}}   \bbh_l^{q}(\Delta)  \F_{l-1}^q
    \|\nonumber \\ \nonumber
   & \leq \left\|
    \bbh_l^{q}(\Delta_n) \bbx_{n,l-1}^q  - \bbh_l^{q}(\Delta_n) \samp_n^{\mathcal{X}} \F_{l-1}^q
    \right\|
  \\ &\qquad +
    \left\|
    \bbh_l^{q}(\Delta_n) \samp_n^{\mathcal{X}} \F_{l-1}^q  - \samp_n^{\mathcal{X}}   \bbh_l^{q}(\Delta)  \F_{l-1}^q
    \right\|
\end{align}
The first term of the last inequality in \eqref{proof_2} can be bounded as $\| \bbx_{n,l-1}^q - \samp_n^{\mathcal{X}}\F_{l-1}^q\|$ with the initial condition $\|\bbx_{n,0}^q - \samp_n^{\mathcal{X}} \F_0^q\|=0$ for $q = 1,\dots,F_0$. Denoting the second term with $D_{l-1}^n$,  and iterating the bounds derived above through layers and features, we obtain:
\begin{equation}\label{eqn:separate_layers}
 \nonumber \|\bPsi(\mathcal{H},\Delta_n,\samp_n^{\mathcal{X}} \F) - \samp_n^{\mathcal{X}} \bPsi(\mathcal{H},\Delta,\F)\|
 \leq
 \sum_{l=0}^L \prod\limits_{l'=l}^L F_{l'} D_l^n.
 \end{equation}
 Therefore, we can focus on each difference term $D_l^n$ and omit the feature and layer indices to simplify the notations. 
 We can write the convolution operation in the spectral domain as
 \begin{align}
    & \nonumber\|\bbh(\Delta_n)\samp_n^{\mathcal{X}} \F - \samp_n^{\mathcal{X}}\bbh(\Delta) \F\| \nonumber \\
   & = \Bigg\| \sum_{i=1}^n \hat{h}(\lambda_i^n) \langle \samp_n^{\mathcal{X}}\F,\bphi_i^n \rangle_{\tm_n}\bphi_i^n \nonumber \\
    & \qquad \qquad \qquad \quad- \sum_{i=1}^\infty \hat{h}(\lambda_i)\langle \F,\bphi_i\rangle_{\tm} \samp_n^{\mathcal{X}} \bphi_i  \Bigg\| 
   \label{eqn:conv-1}
 \end{align}
By adding and subtracting $\sum_{i=1}^{n}\hat{h}(\lambda_i) \langle \samp_n^{\mathcal{X}}\F,\bphi_i^n \rangle_{\tm_n}\bphi_i^n$, by coupling the terms with the same index and using the triangle inequality, we can then write
\begin{align}
        & \nonumber {\Bigg\| \sum_{i=1}^n \hat{h}(\lambda_i^n) \langle \samp_n^{\mathcal{X}}\F,\bphi_i^n \rangle_{\tm_n}\bphi_i^n \nonumber - \sum_{i=1}^\infty \hat{h}(\lambda_i)\langle \F,\bphi_i\rangle_{\tm} \samp_n^{\mathcal{X}} \bphi_i  \Bigg\| }
     \\ 
     &\nonumber \leq  \left\| \sum_{i=1}^n \left(\hat{h}(\lambda_i^n)-\hat {h}(\lambda_i) \right) \langle \samp_n^{\mathcal{X}}\F,\bphi_i^n \rangle_{\tm_n}\bphi_i^n \right\| [T1]\\
     &  +\left\| \sum_{i=1}^n \hat{h}(\lambda_i)\left( \langle \samp_n^{\mathcal{X}}\F,\bphi_i^n \rangle_{\tm_n}\bphi_i^n - \langle \F,\bphi_i\rangle_{\tm} \samp_n^{\mathcal{X}} \bphi_i  \right)  \right\| [T2] \nonumber\\
     & + \Bigg\| \sum_{p=n+1}^\infty \hat{h}(\lambda_p)\langle  \F, \bphi_p\rangle_{\ccalT\ccalM}\samp_n^\ccalX \bphi_p \Bigg\| [T3] \label{eqn:decompose}
\end{align}
We now proceed to prove that $[T1]$ converges to zero in probability as $n$ increases. Fixed a $M_{[T1]} \in \mathbb{N}$, we can always rewrite [T1] as
\begin{align}\label{eqn:t1_M}
  [T1] &= \Bigg\| \sum_{i=1}^{\min\{n,M_{[T1]}\}} \left(\hat{h}(\lambda_i^n)-\hat {h}(\lambda_i) \right) \langle \samp_n^{\mathcal{X}}\F,\bphi_i^n \rangle_{\tm_n}\bphi_i^n \nonumber \\
  & + \sum_{i=M_{[T1]}+1}^{n} \left(\hat{h}(\lambda_i^n)-\hat {h}(\lambda_i) \right) \langle \samp_n^{\mathcal{X}}\F,\bphi_i^n \rangle_{\tm_n}\bphi_i^n \Bigg\|
\end{align}
Please notice that,  when $n<M_{[T1]}$, the last sum is an empty sum. By using the triangle inequality, the orthonormality of the $\bphi_i^n$, the Cauchy-Schwartz inequality $|\langle \samp_n^{\mathcal{X}}\F,\bphi_i^n \rangle_{\tm_n}| \leq \|\samp_n^{\mathcal{X}}\F\|$, and the finiteness of $\|\samp_n^{\mathcal{X}}\F\|$, we can further bound the RHS of \eqref{eqn:t1_M}, obtaining
\begin{align}\label{eqn:t1_M_2}
  [T1] &\leq C_{[T1]}\sum_{i=1}^{\min\{n,M_{[T1]}\}} |\hat{h}(\lambda_i^n)-\hat {h}(\lambda_i)| \nonumber \\
  &+ C_{[T1]}\sum_{i=M_{[T1]}+1}^{n} |\hat{h}(\lambda_i^n)-\hat {h}(\lambda_i)|,
\end{align}
for some constant $C_{[T1]}>0$. At this point, by using the fact that $|a-b| \leq |b|$ and the Lipschitz continuity of $\hat{h}(\cdot)$ (\textbf{A1}), we can further bound the RHS of \eqref{eqn:t1_M_2} as
\begin{align}\label{eqn:t1_M_3}
  [T1] &\leq \underbrace{C_{[T1]}\sum_{i=1}^{\min\{n,M_{[T1]}\}} |\lambda_i^n-\lambda_i|}_{[T1.1]} + \underbrace{C_{[T1]}\sum_{i=M_{[T1]}+1}^{\infty} |\hat{h}(\lambda_i)|}_{[T1.2]}
  \end{align}
 It is clear that we can make  $[T1.2]$ in \eqref{eqn:t1_M_3} arbitrarily small by increasing $M_{[T1]}$ since it is the reminder of a convergent series with positive summands ($\textbf{A2}$). Therefore, for all $\gamma_{[T1]}>0$, we can always choose an $M_{[T1]}$ such that $[T1.2]$ is smaller than $\gamma_{[T1]}/2C$. Fixed $M_{[T1]}$, we can further bound $[T1.1]$ using the spectral convergence result in \eqref{eqn:convergence_spectrum}. In particular, using the definition of limit in probability, letting $0<\gamma_i \leq \gamma_{[T1]}/2CM$, for all $\delta_i>0$, there exist $N_i$ such that for all $n \geq N_i$, it holds
 \begin{gather}
 \label{eqn:eigenvalue}   \mathbb{P}(|\lambda_i^n-\lambda_i|\leq \gamma_i)\geq 1-\delta_i.
 \end{gather}
Therefore, for all $\gamma_{[T1]}>0$ and for all $n \geq \max_i N_i$, it holds
 \begin{align}\label{T11_1}
     [T1.1] \leq C_{[T1]}\sum_{i=1}^{\min\{n,M_{[T1]}\}}\gamma_i\leq \gamma_{[T1]}/2
 \end{align}
 with probability at least $\prod_{i =1}^{\min\{n,M_{[T1]}\}}(1-\delta_i) := 1-\delta_{[T1]}$. This allows us to state that for all $\gamma_{[T1]}>0$, for all $\delta_{[T1]} > 0$, there exist an $N_{[T1]}$ such that, for all $n>N_{[T1]}$, we have
  \begin{gather}
 \label{eqn:conv_T1}   \mathbb{P}([T1]\leq \gamma_{[T1]})\geq 1-\delta_{[T1]},
 \end{gather}
i.e. $[T1]$ converges in probability to zero.
We now proceed to show that $[T2]$ in \eqref{eqn:decompose} converges to zero in probability as $n$ increases. By adding and subtracting $\sum_{i =1}^n\hat{h}(\lambda_i) \langle \samp_n^{\mathcal{X}}\F,\bphi_i^n \rangle_{\tm_n}\samp_n^{\mathcal{X}}\bphi_i$, and by using the triangle inequality, we can write \begin{align}\label{eqn:T2_decompose}
   &\nonumber \left\| \sum_{i=1}^n  \hat{h}(\lambda_i)( \langle \samp_n^{\mathcal{X}} \F,\bphi_i^n \rangle_{\tm_n}\bphi_i^n -  \langle \F,\bphi_i \rangle_{\tm} \samp_n^{\mathcal{X}} \bphi_i )\right\|\\
   &\leq \nonumber \Bigg\|  \sum_{i=1}^n \hat{h}(\lambda_i)\Big(\langle \samp_n^{\mathcal{X}} \F,\bphi_i^n\rangle_{\tm_n}\bphi_i^n  \nonumber \\
   &\qquad \qquad \qquad \qquad - \langle \samp_n^{\mathcal{X}} \F,\bphi_i^n \rangle_{\tm_n} \samp_n^{\mathcal{X}}\bphi_i \Big)\Bigg\| [T2.1]\nonumber \\
   &+ \Bigg\| \sum_{i=1}^n  \hat{h} (\lambda_i)\Big(\langle \samp_n^{\mathcal{X}} \F,\bphi_i^n\rangle_{\tm_n} \samp_n^{\mathcal{X}}\bphi_i \nonumber \\
   &\qquad \qquad\qquad \qquad -\langle \F,\bphi_i\rangle_\tm \samp_n^{\mathcal{X}}\bphi_i \Big) \Bigg\| [T2.2]
\end{align}
We can use now the same approach of $[T1]$. In particular, fixed a $M_{[T2.1]} \in \mathbb{N}$, we can always rewrite $[T2.1]$ and then bound it using the triangle inequality as
\begin{align}\label{eqn:t21_M}
  [T2.1] &= \Bigg\| \sum_{i=1}^{\min\{n,M_{[T2.1]}\}} \hat{h}(\lambda_i)\Big(\langle \samp_n^{\mathcal{X}} \F,\bphi_i^n\rangle_{\tm_n}\bphi_i^n  \nonumber \\
   &\qquad \qquad \qquad \qquad - \langle \samp_n^{\mathcal{X}} \F,\bphi_i^n \rangle_{\tm_n} \samp_n^{\mathcal{X}}\bphi_i \Big) \nonumber \\
  & + \sum_{i=M_{[T2.1]}+1}^{n} \hat{h}(\lambda_i)\Big(\langle \samp_n^{\mathcal{X}} \F,\bphi_i^n\rangle_{\tm_n}\bphi_i^n  \nonumber \\
   &\qquad \qquad \qquad \qquad - \langle \samp_n^{\mathcal{X}} \F,\bphi_i^n \rangle_{\tm_n} \samp_n^{\mathcal{X}}\bphi_i \Big)\Bigg\| \nonumber \\
   &\leq \Bigg\| \sum_{i=1}^{\min\{n,M_{[T2.1]}\}} \hat{h}(\lambda_i)\Big(\langle \samp_n^{\mathcal{X}} \F,\bphi_i^n\rangle_{\tm_n}\bphi_i^n  \nonumber \\
   &\qquad \qquad \qquad \qquad - \langle \samp_n^{\mathcal{X}} \F,\bphi_i^n \rangle_{\tm_n} \samp_n^{\mathcal{X}}\bphi_i \Big)\Bigg\| \nonumber \\
   &+ \Bigg \|\sum_{i=M_{[T2.1]}+1}^{n} \hat{h}(\lambda_i)\Big(\langle \samp_n^{\mathcal{X}} \F,\bphi_i^n\rangle_{\tm_n}\bphi_i^n  \nonumber \\
   &\qquad \qquad \qquad \qquad - \langle \samp_n^{\mathcal{X}} \F,\bphi_i^n \rangle_{\tm_n} \samp_n^{\mathcal{X}}\bphi_i \Big)\Bigg\| 
\end{align}
We can now further bound the RHS of \eqref{eqn:t21_M} by using the triangle inequality, the Cauchy-Schwarz inequality $|\langle \samp_n^{\mathcal{X}}\F,\bphi_i^n \rangle_{\tm_n}| \leq \|\samp_n^{\mathcal{X}}\F\|$ , the non-amplifying frequency response (\textbf{A1}, for the first term), the finiteness of $\|\samp_n^{\mathcal{X}}\F\|$, and the finiteness of $\|\bphi_i^n-\samp_n^{\mathcal{X}}\bphi_i\|$ (for the second term) as
\begin{align}\label{eqn:t21_decom}
  [T2.1] &\leq \underbrace{C_{[T2.1]} \sum_{i=1}^{\min\{n,M_{[T2.1]}\}} \|\bphi_i^n-\samp_n^{\mathcal{X}}\bphi_i\|}_{[T2.1.1]} \nonumber \\
  & + \underbrace{C_{[T2.1]} \sum_{i=M_{[T2.1]}+1}^\infty |\hat{h}(\lambda_i)|}_{[T2.1.2]},
\end{align}
for some constant $C_{[T2.1]}>0$. Leveraging the same arguments we used for $[T1.1]$ and $[T1.2]$ in 
\eqref{eqn:t1_M_3} to bound $[T2.1.1]$ and $[T2.1.2]$ in 
\eqref{eqn:t21_decom}, respectively, but using the convergence of the eigenvectors and not of the eigenvalues from \eqref{eqn:convergence_spectrum}, we can state that for all $\gamma_{[T2.1]}>0$, for all $\delta_{[T2.1]} > 0$, there exist an $N_{[T2.1]}$ such that, for all $n>N_{[T2.1]}$, we have
  \begin{gather}
 \label{eqn:conv_T21}   \mathbb{P}([T2.1]\leq \gamma_{[T2.1]})\geq 1-\delta_{[T2.1]},
 \end{gather}
i.e. $[T2.1]$ converges in probability to zero. Following the same procedure we used to obtain the bound in \eqref{eqn:t21_M} for $[T2.1]$, we can obtain the following bound for $[T2.2]$:
\begin{align}\label{eqn:t22_M}
  [T2.2] &\leq \Bigg\| \sum_{i=1}^{\min\{n,M_{[T2.2]}\}}\hat{h} (\lambda_i)\Big(\langle \samp_n^{\mathcal{X}} \F,\bphi_i^n\rangle_{\tm_n} \samp_n^{\mathcal{X}}\bphi_i \nonumber \\
   &\qquad \qquad\qquad \qquad -\langle \F,\bphi_i\rangle_\tm \samp_n^{\mathcal{X}}\bphi_i \Big)\Bigg\| \nonumber \nonumber \\
   &+ \Bigg \|\sum_{i=M_{[T2.2]}+1}^{n} \hat{h} (\lambda_i)\Big(\langle \samp_n^{\mathcal{X}} \F,\bphi_i^n\rangle_{\tm_n} \samp_n^{\mathcal{X}}\bphi_i \nonumber \\
   &\qquad \qquad\qquad \qquad -\langle \F,\bphi_i\rangle_\tm \samp_n^{\mathcal{X}}\bphi_i \Big)\Bigg\| 
\end{align}
We further bound the RHS of \eqref{eqn:t22_M} by using the triangle and Cauchy-Schwarz inequalities, the non-amplifying frequency response (for the first term),  the finiteness of $\|\samp_n^{\mathcal{X}}\F\|$ and $\|\F\|$, and the finiteness of $\|\bphi_i^n-\samp_n^{\mathcal{X}}\bphi_i\|$ (for the second term), as
\begin{align}\label{eqn:t22_decom}
  [T2.2] &\leq \underbrace{C_{[T2.2]} \sum_{i=1}^{\min\{n,M_{[T2.2]}\}} |\langle \samp_n^{\mathcal{X}} \F,\bphi_i^n\rangle_{\tm_n} -\langle \F,\bphi_i\rangle_\tm|}_{[T2.2.1]} \nonumber\\
  & + \underbrace{C_{[T2.2]} \sum_{i=M_{[T2.2]}+1}^\infty |\hat{h}(\lambda_i)|}_{[T2.2.2]},
\end{align}
for some constant $C_{[T2.2]}>0$.  It is trivial, from the weak law of large numbers and from \eqref{emp_metr}-\eqref{emp_metr_versions}, that 
 \begin{equation}\label{eqn:conv_inn_1}
\lim_{n\rightarrow \infty}\left|\langle \samp_n^{\mathcal{X}} \F,\samp_n^{\mathcal{X}}\bphi_i\rangle_{\tm_n}  -\langle \F,\bphi_i \rangle_\tm\right| = 0,
\end{equation}
with the limit taken in probability. By direct substitution and using the distributive law of the dot product, we can write
\begin{align}
    &\left|\langle \samp_n^{\mathcal{X}} \F,\samp_n^{\mathcal{X}}\bphi_i\rangle_{\tm_n} - \langle \samp_n^{\mathcal{X}} \F,\bphi_i^n\rangle_{\tm_n}\right| \nonumber \\
    & = \left| \frac{1}{n}\sum_{i=1}^n \left(\samp_n^{\mathcal{X}} \F(x_i)\dotp\samp_n^{\mathcal{X}}\bphi_i(x_i)-\samp_n^{\mathcal{X}} \F(x_i)\dotp\bphi_i^n(x_i)\right)\right| \nonumber \\
    & = \left|\langle \samp_n^{\mathcal{X}} \F,\samp_n^{\mathcal{X}}\bphi_i-\bphi_i^n\rangle_{\tm_n}\right|  \leq \|\samp_n^{\mathcal{X}} \F\| \|\samp_n^{\mathcal{X}}\bphi_i-\bphi_i^n\|,
\end{align}
where the last inequality is obtained using the Cauchy-Schwartz inequality. Therefore, using again the spectral convergence of eigenvectors from \eqref{eqn:convergence_spectrum}, we can write
\begin{equation}\label{eqn:conv_inn2}
\lim_{n \rightarrow \infty}\left|\langle \samp_n^{\mathcal{X}} \F,\samp_n^{\mathcal{X}}\bphi_i\rangle_{\tm_n} - \langle \samp_n^{\mathcal{X}} \F,\bphi_i^n\rangle_{\tm_n}\right| = 0,
\end{equation}
where the limit is taken in probability.
As a direct consequence of the \eqref{eqn:conv_inn_1} and \eqref{eqn:conv_inn2}, we can directly state that
\begin{equation}\label{eqn:conv_inn_fin}
\lim_{n \rightarrow \infty}\left|\langle \samp_n^{\mathcal{X}} \F,\bphi_i^n\rangle_{\tm_n}-\langle  \F,\bphi_i\rangle_\tm\right| = 0,
\end{equation}
again with the limit in probability. At this point, leveraging the same arguments we used for $[T1.1]$ and $[T1.2]$ in 
\eqref{eqn:t1_M_3} (and for $[T2.1.1]$ and $[T2.1.2]$ in \eqref{eqn:t21_decom})  to bound $[T2.2.1]$ and $[T2.2.2]$ in 
\eqref{eqn:t22_decom}, respectively, but using the convergence of the inner products in \eqref{eqn:conv_inn_fin}, we can state that for all $\gamma_{[T2.2]}>0$, for all $\delta_{[T2.2]} > 0$, there exist an $N_{[T2.2]}$ such that, for all $n>N_{[T2.2]}$:
  \begin{gather}
 \label{eqn:conv_T22}   \mathbb{P}([T2.2]\leq \gamma_{[T2.2]})\geq 1-\delta_{[T2.2]},
 \end{gather}
i.e. $[T2.2]$ converges in probability to zero. As a consequence, we can state that for all $\gamma_{[T2]}>0$, for all $\delta_{[T2]} > 0$, there exist an $N_{[T2]}$ such that, for all $n>N_{[T2]}$, we have
  \begin{gather}
 \label{eqn:conv_T2}   \mathbb{P}([T2]\leq \gamma_{[T2]})\geq 1-\delta_{[T2]},
 \end{gather}
i.e. $[T2]$ converges in probability to zero. We are now missing only the convergence in probability of $[T3]$ from \eqref{eqn:decompose}. However, $[T3]$ is again the reminder of a convergent series with positive summands ($\textbf{A2}$), implying that it deterministically goes to zero as $n$ increases. Therefore, for all $\gamma_{[T3]}>0$, there exist an $N_{[T3]}$ such that, for all $n>N_{[T3]}$, we have
  \begin{gather}
 \label{eqn:conv_T3}   [T3]\leq \gamma_{[T3]}
 \end{gather}

As a direct consequence of
\eqref{eqn:conv_T1}-\eqref{eqn:conv_T2}-\eqref{eqn:conv_T3}, we can state that for all $\gamma > 0$, for all $\delta>0$, there exist a $N$ such that, for all $n>N$, we have
  \begin{gather}
 \label{eqn:conv_T}   \mathbb{P}([T1]+[T2]+[T3]\leq \gamma)\geq 1-\delta,
 \end{gather}
Combining \eqref{eqn:conv_T} with \eqref{eqn:decompose}, we can finally state that
\begin{equation}
    \lim_{n \rightarrow \infty }D_l^n=\lim_{n \rightarrow \infty }\|\bbh(\Delta_n)\samp_n^{\mathcal{X}} \F - \samp_n^{\mathcal{X}}\bbh(\Delta) \F\| = 0,
\label{eqn:final_conv}
\end{equation}
where the limit is taken in probability. The proof is concluded by combining \eqref{eqn:final_conv} and \eqref{eqn:separate_layers}.

\bibliographystyle{IEEEtran}
\bibliography{refs}
\clearpage
\setcounter{page}{1}
\twocolumn[%
   \begin{center}
     {\huge Supplemental Materials }\\
     \end{center}\vspace{0.5cm}
]
\vspace{1cm}
\setcounter{subsection}{0}
\subsection{Sheaf Laplacian Algorithms}\label{ap:algos}
\begin{algorithm}[H]
\footnotesize
   \caption{: Local PCA \cite{singer2012vdm}}    \hspace*{\algorithmicindent} \textbf{Inputs}: \\
    \hspace*{\algorithmicindent} \quad $\mathcal{X} \subset \mathbb{R}^p$: Manifold samples. \vspace{.05cm}\\
    \hspace*{\algorithmicindent} \quad $\epsilon_{\textrm{PCA}} > 0$: Scale parameter \vspace{.05cm}\\
    \hspace*{\algorithmicindent} \quad $K(\cdot) \in C^2(\mathbb{R})$: positive monotonic  supported on $[0, 1]$ \vspace{.05cm}\\
    \hspace*{\algorithmicindent} \textbf{Outputs}: \vspace{.05cm}\\
    \hspace*{\algorithmicindent} \quad $\{\mathbf{O}_i\}_{x_i \in \mathcal{X}}$: Orthogonal transformation
    \begin{algorithmic}[1]
        \Function{LOCAL PCA\,} {\textbf{Inputs}}
                \For{$x_i \in \mathcal{X}$}
                    \State Compute $\mathcal{N}^{\textrm{P}}_i = \{x_j: 0 < \|x_i - x_j\|_{\mathbb{R}^p}\leq \sqrt{\epsilon_{\textrm{PCA}}}\}$\vspace{.05cm}
                    \State Compute $\mathbf{X}_i = [\dots, x_i - x_j, \dots]$, $x_j \in \mathcal{N}^{\textrm{P}}_i$ \vspace{.05cm}
                    \State Compute $\mathbf{C}_i$ with $[\mathbf{C}_i]_{j,j}=\sqrt{K\big(\frac{||x_i -x_j||}{\sqrt{\epsilon_{\textrm{PCA}}}}\big)}$ \vspace{.05cm}
                    \State Compute $\mathbf{B}_i=\mathbf{X}_i\mathbf{C}_i$ and $\mathbf{R}_i=\mathbf{B}_i^T\mathbf{B}_i$  \vspace{.05cm}
                    \State Eigendecompose $\mathbf{R}_i=\mathbf{M}_i\Sigma_i\mathbf{J}_i^T$ \vspace{.05cm}
                \EndFor
        \State \textbf{end}
        \State Compute $\hd$ as in \cite{singer2012vdm} \vspace{.05cm}
        \State Set $\mathbf{O}_i$ to be the first $\hd$ columns of $\mathbf{M}_i$ \vspace{.05cm}\\
        \Return: \\
        \hspace*{\algorithmicindent} \quad $\{\mathbf{O}_i\}_{x_i \in \mathcal{X}}$ 
       \EndFunction
\end{algorithmic}
\end{algorithm}\label{algo:locpca}

\begin{algorithm}[H]
\footnotesize
   \caption{: Sheaf Laplacian \cite{singer2012vdm}}   \hspace*{\algorithmicindent} \textbf{Inputs}: \\
    \hspace*{\algorithmicindent} \quad $\mathcal{X} \subset \mathbb{R}^p$: Manifold samples.  \vspace{.05cm}\\
    \hspace*{\algorithmicindent} \quad $\epsilon_n > 0$: Scale parameter for geometric graph \vspace{.05cm}\\
    \hspace*{\algorithmicindent} \quad $\epsilon_{\textrm{PCA}} > 0$: Scale parameter for local PCA \vspace{.05cm}\\
    \hspace*{\algorithmicindent} \quad $K(\cdot) \in C^2(\mathbb{R})$: positive monotonic  supported on $[0, 1]$\\
    \hspace*{\algorithmicindent} \textbf{Outputs}: \vspace{.05cm}\\
    \hspace*{\algorithmicindent} \quad $\{\mathbf{O}_i\}_{x_i \in \mathcal{X}}$: Orthogonal transformation \vspace{.05cm}\\
    \hspace*{\algorithmicindent} \quad $\Delta_n$: Normalized Sheaf Laplacian
    \begin{algorithmic}[1]
        \Function{SHEAF LAPLACIAN\,} {\textbf{Inputs}}
        \State Compute graph $\mathcal{M}_n$ with edge weights as in \eqref{graph_weights}
                \For{$x_i \in \mathcal{X}$}
                    \State Compute $\{\mathbf{O}_i\}$ with Algorithm 2
                \EndFor
                \State \textbf{end}
                \For{$x_i \in \mathcal{X}$}
                     \For{$x_j \in \mathcal{N}^{\textrm{P}}_i$}
                        \State Compute $\widetilde{\mathbf{O}}_{i,j}=\mathbf{O}_i^T\mathbf{O}_j\svdeq \mathbf{M}_{i,j}\Sigma_i\mathbf{V}_{i,j}^T$ \vspace{.05cm}
                        \State Compute $\mathbf{O}_{i,j}=\mathbf{M}_{i,j}\mathbf{V}_{i,j}^T$ \vspace{.05cm}
                        \State Compute $\textrm{deg}(i)=\sum_j w_{i,j}$ \vspace{.05cm}
                        \State Compute $\textrm{ndeg}(i)=\sum_j \frac{w_{i,j}}{\textrm{deg}(i)\textrm{deg}(j)}$ \vspace{.05cm}
                        \State Compute $\widetilde{\mathbf{D}}_i=\textrm{deg}(i)\mathbf{I}_{\hd}$ and  $\mathbf{D}_i=\textrm{ndeg}(i)\mathbf{I}_{\hd}$ \vspace{.05cm}
                        \State Compute $\mathbf{S}_{i,j}=w_{i,j}\widetilde{\mathbf{D}}^{-1}_i\mathbf{O}_{i,j}\widetilde{\mathbf{D}}^{-1}_i$ \vspace{.05cm}
                    \EndFor
                    \State \textbf{end}
                \EndFor
        \State \textbf{end}
        \State Compute block matrix $\mathbf{S}$ with $\mathbf{S}_{i,j}$s as blocks \vspace{.05cm}
        \State Compute block diagonal matrix $\mathbf{D}$ with $\mathbf{D}_{i,i}$ as blocks 
        \State Compute $\Delta_n = \epsilon_n^{-1}\big(\D^{-1}\S - \mathbf{I}\big)$ \vspace{.05cm}\\
        \Return: \\
        \hspace*{\algorithmicindent} \quad $\{\mathbf{O}_i\}_{x_i \in \mathcal{X}}$, $\Delta_n$
       \EndFunction
\end{algorithmic}
\end{algorithm}\label{algo:sheaflap}
\newpage
\subsection{Proof of Proposition 1} \label{ap:prop1}
\begin{proof}[Proof of Proposition  \ref{prop:parametric-filter}]
By definition of frequency representation in \eqref{freq_resp} we have:
\begin{align}
\label{g_freq}
    &\big[\hat{G}\big]_i = \langle \G, \Phii \rangle = \int_{\M}\langle \G(x), \Phii(x) \rangle_{\tmx} \textrm{d}\mu(x)
\end{align}
Injecting \eqref{param_conv} in \eqref{g_freq}, we get:
\begin{align}
\label{g_freq_f}
    &\big[\hat{G}\big]_i = \langle \int_0^{\infty}\th(t)e^{t\Delta}\F(x)\textrm{d}t, \Phii \rangle 
\end{align}
For the linearity of integrals and inner products, we can write:
\begin{align}
\label{g_freq_int_out}
    &\big[\hat{G}\big]_i = \int_0^{\infty}\th(t)\langle e^{t\Delta}\F(x), \Phii \rangle \textrm{d}t
\end{align}
Finally, exploiting first the self-adjointness of $\Delta$ and then the eigenvector fields definition in \eqref{eigen}, we can write:
\begin{align}
\label{g_freq_final}
    \big[\hat{G}\big]_i &= \int_0^{\infty}\th(t)\langle e^{t\Delta}\F(x), \Phii \rangle \textrm{d}t \nonumber\\
    &=\int_0^{\infty}\th(t)\langle \F(x), e^{t\Delta}\Phii \rangle \textrm{d}t \nonumber \\  
    &=\int_0^{\infty}\th(t)\langle \F(x), e^{-t\lambda_i}\Phii \rangle \textrm{d}t \nonumber \\
    &=\int_0^{\infty}\th(t)e^{-t\lambda_i}\langle \F(x), \Phii \rangle \textrm{d}t,
\end{align}
which concludes the proof.
\end{proof}
\end{document}